\documentclass{article}

\usepackage{ucs}
\usepackage[utf8]{inputenc}
\usepackage{amsmath}
\usepackage{amsfonts}

\usepackage{amssymb}%
\usepackage{amsthm}
\usepackage{stmaryrd}
\usepackage{tikz}
\usetikzlibrary{calc}
\usetikzlibrary{shapes}
\usetikzlibrary{arrows}
\usepackage[noend]{algpseudocode}
\usepackage{algorithm}
\usepackage{xcolor}
\usepackage{graphicx}
\usepackage{multirow}
\usepackage{changepage}
\usepackage[hidelinks]{hyperref}
\usepackage{xspace}
\usepackage{caption}
\usepackage{subcaption}
\usepackage{enumitem}

\usepackage{mathrsfs}%
\usepackage[title]{appendix}%
\usepackage{manyfoot}%

\usepackage{url}

\newtheorem{definition}{Definition}
\newtheorem{theorem}{Theorem}
\newtheorem{lemma}{Lemma}

\newtheorem{fact}{Fact}

\newtheorem{example}{Example}
\newtheorem*{remark}{Remark}
\newtheorem{problem}{Problem}
\usepackage{todonotes}
\tikzset{source/.style={draw, regular polygon, regular polygon sides = 3, minimum height=0.3cm}}	
\tikzset{switch/.style={draw, circle, minimum height=0.3cm}}
\tikzset{puits/.style={draw, regular polygon, regular polygon sides = 4, minimum height=0.3cm}}

\newcommand{\valid}{VALID\xspace}

\newcommand{\maxminm}{$\max$-m\xspace}
\newcommand{\minmaxm}{$\min$-M\xspace}
\newcommand{\minr}{$\min$-R\xspace}

\newcommand{\round}[1]{\widetilde{#1}\xspace}

\newcommand{\flowst}[1]{F_{#1}}
\newcommand{\roundedflowst}[1]{\round{F}_{#1}}
\newcommand{\flow}[2]{\flowst{#1}(#2)}
\newcommand{\roundedflow}[2]{\roundedflowst{#1}(#2)}
\newcommand{\load}[2]{load_{#1}(#2)}
\newcommand{\roundedload}[2]{\round{load}_{#1}(#2)}
\newcommand{\minload}[1]{m(#1)\xspace}
\newcommand{\roundedminload}[1]{\round{m}(#1)\xspace}
\newcommand{\maxload}[1]{M(#1)\xspace}
\newcommand{\roundedmaxload}[1]{\round{M}(#1)\xspace}
\newcommand{\loadreserve}[1]{R(#1)\xspace}
\newcommand{\roundedloadreserve}[1]{\round{R}(#1)\xspace}

\newcommand{\productionfunctionname}{Prod}
\newcommand{\capacityfunctionname}{Cap}
\newcommand{\powerfunctionname}{Pow}
\newcommand{\production}[1]{\productionfunctionname(#1)}
\newcommand{\capacity}[1]{\capacityfunctionname(#1)}
\newcommand{\power}[1]{\powerfunctionname(#1)}

\newcommand{\mininputflowst}{i}
\newcommand{\maxoutputflowst}{o}
\newcommand{\mininputflow}[5]{\mininputflowst(#1, #2, #3, #4, #5)}
\newcommand{\maxoutputflow}[5]{\maxoutputflowst(#1, #2, #3, #4, #5)}
\newcommand{\mininputflowsm}[3]{\mininputflowst(#1, #2, #3)}
\newcommand{\maxoutputflowsm}[3]{\maxoutputflowst(#1, #2, #3)}
\begin{document}
	
 \title{Configuring an heterogeneous smartgrid network: complexity and approximations for tree topologies.\footnote{This is a preprint version of \url{https://doi.org/10.1007/s10898-023-01338-0}}}
 
 \author{Dominique Barth \footnote{\url{dominique.barth@uvsq.fr}, \url{thierry.mautor@uvsq.fr}, DAVID, University of Versailles-St Quentin, 45, avenue des États-Unis, Versailles, 78035, France} \and Thierry Mauto $^\dagger$ \and Dimitri Watel \footnote{\url{dimitri.watel@ensiie.fr}, ENSIIE, 1, square de la Résistance, Evry, 91000, France}~\footnote{SAMOVAR, 9 Rue Charles Fourier, Evry, 91000, France} \and Marc-Antoine Weisser \footnote{\url{marc-antoine.weisser@centralesupelec.fr}, LISN, Centrale Supélec, Paris-Saclay University, Rue Noetzlin, Gif-sur-Yvette, 91190, France}}
 
 \date{}

\maketitle

\vspace{-0.5cm}
\begin{abstract}
	We address the problem of configuring a power distribution network with reliability and resilience objectives by satisfying the demands of the consumers and saturating each production source as little as possible. We consider power distribution networks containing source nodes producing electricity, nodes representing electricity consumers and switches between them. Configuring this network consists in deciding the orientation of the links between the nodes of the network. The electric flow is a direct consequence of the chosen configuration and can be computed in polynomial time. It is valid if it satisfies the demand of each consumer and capacity constraints on the network. In such a case, we study the problem of determining a feasible solution that balances the loads of the sources, that is their production rates. 
	We use three metrics to measure the quality of a solution: minimizing the maximum load, maximizing the minimum load and minimizing the difference of the maximum and the minimum loads. This defines optimization problems called  respectively \minmaxm, \maxminm and \minr.
	
	In the case where the graph of the network is a tree, it is known that the problem of building a valid configuration is polynomial. We show the three optimization variants have distinct properties regarding the theoretical complexity and the approximability. Particularly, we show that \minmaxm is polynomial, that \maxminm is NP-Hard but belongs to the class FPTAS and that \minr is NP-Hard, cannot be approximated to within any exponential relative ratio but, for any $\varepsilon > 0$, there exists an algorithm for which the value of the returned solution equals the value of an optimal solution shifted by at most $\varepsilon$.
\end{abstract}

\textit{Complexity, Approximability, Electrical network flow, Tree topology}

	\section{Introduction}
	Smart electrical distribution networks are becoming increasingly heterogeneous \cite{Frieden2021} with a multiplication 
	% and geographical distribution 
	of sources of production of different natures and capacities (wind, solar, biomass, etc.)
	with a diverse geographical distribution in particular at the scale of an urban territory. The configuration of these resilient networks must be dynamically optimized \cite{Calhau2019} by guaranteeing a reliability of service 
	%whatever the evolution of the conditions concerning not only the capacity of heterogeneous sources and the demands of consumers but also breakdowns or degradation of meteorological  conditions 
	taking into account changes in the capacity of sources or in consumer demands but also breakdowns or degradations due for example to weather conditions \cite{Devaux2014,Ghiani05}. This dynamic configuration of the distribution network is essential for the current development of smartgrids or microgrids \cite{Sahua2022} implementing collective self-consumption grids interconnecting on a local territory a set of prosumers (i.e. consumers able to produce, store and therefore share energy) \cite{Barth2022a,Cejka2021,Frieden2021}. 
	
	% The configuration of such a reliable smartgrid must therefore be able to manage a failure or a variation in the behavior of local sources or consumers. 
	Configuring such a distribution network means deciding which components (lines, sources, switches) should be activated or not. Reliability can be defined as the capacity of the electrical system to supply electricity in the quantity and with the quality demanded by consumers by using available sources reliably and fairly. Thus, to guarantee reliability, the objective here is to find the configuration satisfying the   demands of the consumers \cite{Calhau2019,Moslehi2010,Tang2014} by saturating as little as possible the production capacity of (heterogeneous) sources and the flow capacity of links and switches \cite{Barth2021,Blasi2022}.
	
	Configuration problems with reliability and resilience objectives are often considered from a graph theory perspective \cite{Atkins2009,Quiros-Tortos2013}. Thus, some solutions seek balanced configurations in terms of load power \cite{Shen2018} and other ones configure the network in disjoint balanced subnetworks \cite{Guo2016,Li2010,Sahua2022,Tang2014}, in particular in the form of spanning trees or sub-DAGS \cite{Li2014}. 	
	
	But, as the electric flow in a network is a direct consequence of the chosen configuration \cite{Shen2018}, the objective is more to determine if there exists a configuration satisfying the production and capacity constraints and the consumption demands than to compute an electric flow in a graph (such as in \cite{Christiano2011}). Moreover, the reliability of a network can be guaranteed by the existence of a configuration which does not use all the capacity of each link and each switch in order to contain the impact of the snowball effect during breakdowns. So given a desired maximum percentage $s$, the problem we solve allows us to determine if there is a network configuration using each switch and each link at most $s$ percents of its capacity. This percentage is called the load of the node. In \cite{Barth2021,Barth2022} we prove that this existence problem is NP-complete for general network topologies but polynomial for trees. 
	
	This paper introduces an optimization perspective to the problem so that the objective is now to determine a feasible configuration of the network where the loads of the sources are as balanced as possible. Such a balance responds, on the one hand, to the need to avoid saturating sources that would not be able to adapt to failures or sudden changes in demand (to avoid snowball effects), and, on the other hand, to propose a fair use of the energy produced by each actor in the context of a collective self-consumption grid, for example	\cite{Cejka2021,Frieden2021}.
	 
	We therefore introduce three metrics that can answer this problematic. As done, in \cite{Barth2021}, we introduce the load reserve, which is the difference between the maximum and the minimum loads of the sources. This load reserve should be minimized. Another way to balance the network is also to maximize the minimum load or minimize the maximum load. We show in this paper that, when the network topology is a tree, the three problems have distinct properties regarding the theoretical complexity and the approximability. Minimizing the maximum load is polynomial, which is a direct consequence of the result of \cite{Barth2022}, while minimizing the load reserve and maximizing the minimum load are strongly NP-hard problems. The two latter problems differ in their approximability. There exists an FPTAS when the objective is to maximize the minimum load, but no polynomial approximation (with any exponential ratio) exists for the other problem. Finally, for the two problems, there exists an FPTAS with an absolute ratio, i.e., for any $\varepsilon > 0$, there exists an algorithm for which the value of the returned solution is equal to the value of an optimal solution shifted by at most $\varepsilon$. 
	
	The FPTAS algorithms are based on the fact that we can manipulate the flows with a small precision. An interesting property is that, if we consider the physical electrical network, we cannot (and we do not need) to do the calculations with a flow that has an infinite precision, and this limitation naturally leads to an approximation algorithm that can run in polynomial time. 
	
	The paper is organized as follows. First, we define our model and the related computational problems. We then prove the polynomial and NP-hardness results, and end with the approximability result.
	
	\section{Model}
	
	The electrical distribution network is represented by a graph  $G=(V,E)$  in which there are three types of vertices $S$, $W$ and $P$. $S$ represents the set of {\em sources} of the network; each vertex $s \in S $ is characterized by a maximum supply capacity denoted $\production{s} \in \mathbb{N}^*$. $W$ corresponds to the set of {\em switches};  each vertex $ w \in W $ has a maximum flow capacity $\capacity{w} \in \mathbb{N}$. And $P$ represents the set of {\em sinks} (or consumers); each vertex $ p \in P $ is characterized by a demand $\power{p} \in \mathbb{N}$. The edges are the connections between the vertices. In the following, we focus on graphs $G$ that are trees.
	
	\begin{remark}
		Note that we do not require the sources or the sinks to be leaves of the tree. Similarly, we do not require the switches to be internal nodes. 
	\end{remark}
	
	An \textbf{orientation} $\mathcal{O}$ of a graph $G$ is a function associating each edge $ [u, v] \in E $ with a couple  $ (u, v) $ or $ (v, u) $ corresponding to the orientation of this edge. The directed graph obtained with $\mathcal{O}$ must be a DAG to avoid circuits in the graph of the electrical flow. In this paper we only focus on the case where $G$ is a tree, so the DAG constraint is necessarily satisfied.
	
	Note that for a given orientation, a switch is \emph{activated} only if it belongs to at least one induced path from a source to a sink. The activation of the switches can thus be seen as a consequence of the orientation of the edges. Consequently, in the rest of the paper, only the orientation will be considered to determine the configuration of a network. 
	
	Note also that even if the use of expensive technologies based on diodes forbidding one direction of electric flow   on a link is sometimes considered, they are not deployed in large electrical distribution networks. This orientation of each link of the graph is a consequence of the positions of sources and sinks on the network and their demand and production values, with the consequence that a brutal variation on production or demand can cause blackout situations (as the northeast one in 2023). In our model however, we consider that the orientation of the links is a decision variable (as for example in the case of restoration methods based on the search for the shortest paths considered as oriented from the source to the sinks \cite{Labrini2015,Quiros-Tortos2013}). On the one hand, this avoids  having to complicate the graph model by electrical considerations and on the other hand having to calculate a flow at each step of the optimization algorithms. In practice, with the considerations linked to the activation of the switches indicated above, the optimal solutions finally obtained turn out to be a realistic approximation of the real behavior of an electrical network in terms of orientation of the links, which is particularly true in the case of trees \cite{Hong2017}. Thus, an orientation of the arcs can therefore be considered valid from the point of view of an electric flow if for each arc (u,v), the total quantity of electricity in u is greater than the one in v, which will always be the case in the configurations discussed. 
	
	\subsection{Flow in the oriented tree and feasible orientation}
	
	Given an orientation $\mathcal{O}$, the network $G$ becomes a flow network with capacity only on the switches, with multiple sources and sinks. Such a flow network can be made compatible with the one used in  \cite{Christiano2011}, but the fundamental difference lies in the constraints satisfied by the flow. In addition to the classical conservation and capacity constraints, the (electric) flow entering each node in $W \cup P$ must be equitably distributed over all its incoming arcs (i.e., the arcs are considered equivalent to resistors with same resistance value). In the special case of a source, the same constraint applies except that the source production itself counts as an additional fictive arc. Consequently, the goal here is then not to determine a maximum flow since it is unique and determined by the orientation: since the directed graph we obtain with this orientation is a DAG, we can compute it by going up from $P$. Our objective is to know if there exists an orientation for which the implied unique flow satisfies the demands of the nodes of $P$ and respects the supply capacities of the sources of $S$ and the capacities of the switches of $W$. Figure~\ref{fig:example} gives examples of an orientation and the flow induced by that orientation. 
	
	\begin{definition}
		\label{def:valid:3}
		Let $\mathcal{O}$ be an orientation of the edges of the graph $G$. Let $\Gamma^+(v)$ and $d^-(v)$ be respectively the set of successors and the in-degree of $v$, the flow $\flow{\mathcal{O}}{u, v}$ going through each arc $(u, v)$ entering $v$ is  
		$$\begin{matrix}
			\text{ if }& v \in W, & v \in S, & v \in P\\
			\\
			\flow{\mathcal{O}}{u, v} = & \dfrac{\sum\limits_{w \in \Gamma^+(v)} \flow{\mathcal{O}}{v, w}}{d^-(v)}, 
			&\dfrac{\sum\limits_{w \in \Gamma^+(v)} \flow{\mathcal{O}}{v, w}}{d^-(v) + 1},
			&\dfrac{\sum\limits_{w \in \Gamma^+(v)} \flow{\mathcal{O}}{v, w} + \power{v}}{d^-(v)}
		\end{matrix}$$
	
		Since $u$ does not intervene in the formula, we also denote this value by $\flow{\mathcal{O}}{v}$.
		
		We say the orientation $\mathcal{O}$ is feasible if:
		\begin{itemize}
			\item it satisfies \emph{the demand constraints}: for every switch $v \in W$ and every sink $v \in P$, we have $d^-(v) > 0$.
			\item it satisfies \emph{the capacity constraints}: $\flow{\mathcal{O}}{v} \cdot d^-(v) \leq \capacity{v}$  for $v \in W$ and $\flow{\mathcal{O}}{v} \leq \production{v}$ for $v \in S$. Note that these capacity constraints include a supply capacity constraint.
		\end{itemize}
	\end{definition}

%\begin{remark}
%	The demand constraints impose every switch and every sink to have an entering arc. One could argue that a switch not having an input arc is reasonable as long as the output flow of the switch is 0. Similarly for a sink with a null called power. If such an orientation $\mathcal{O}$ exists, then we can build a feasible  orientation $\mathcal{O}'$ satisfying $\flow{\mathcal{O}'}{v} = \flow{\mathcal{O}}{v}$ for all the nodes. Let $V0$ be a maximal set of nodes with a null output flow. 
%\end{remark}
	
	\begin{example}
		Assuming the input network is the one shown in Figure~\ref{fig:example:a}. We give 5 examples of orientations, two of which satisfy the demand and the capacity constraints. Note that sources may not be leaves and may have entering arcs (as in Figure~\ref{fig:example:c}). 
		
		\begin{figure}[!ht]
			\captionsetup[subfigure]{justification=centering}
			\center
			\begin{subfigure}[t]{0.4\textwidth}
				\begin{tikzpicture}
					
					\clip (-3.75,-1.5) rectangle (0.75, 2);
					
					\node[source] [label=above:$s_1$, label=left:{\small100}] (S1) at (-3,1) {};
					\node[source] [label=above:$s_2$, label=right:{\small20}] (S2) at (0,1) {};
					
					\node[switch] [label=above left:$w_1$, label=left:{\small$60$}] (W1) at (-3,0) {};
					\node[switch] [label=above:$w_2$, label=below:\small{20}] (W12) at (-1.5,0) {};
					\node[switch] [label=above right:$w_3$, label=right:{\small$35$}] (W2) at (0,0) {};
					
					\node[puits] [label=below:$p_1$, label=left:{\small$50$}] (P1) at (-3,-1) {};
					\node[puits] [label=below:$p_{2}$, label=right:{\small$20$}] (P2) at (0,-1) {};
					\node[puits] [label=above:$p_{3}$, label=right:{\small$10$}] (P3) at (-1.5,1) {};
					
					\draw[-, >=latex] (S1) -- (W1);
					\draw[-, >=latex] (S1) -- (P3);
					\draw[-, >=latex] (S2) -- (W2);
					
					\draw (W1) -- (P1);
					\draw (W12) -- (W1);
					\draw (W12) -- (W2);
					\draw (W2) -- (P2);
				\end{tikzpicture}
				\caption{The instance, each node is associated with a production, capacity or power depending if it is a source, a switch or a sink.}	
				\label{fig:example:a}	
			\end{subfigure}%
			\begin{subfigure}[t]{0.4\textwidth}
				\begin{tikzpicture}
					\clip (-3.75,-1.5) rectangle (0.75, 2);
					
					\node[source] [label=above:$60/100$] (S1) at (-3,1) {};
					\node[source] [label=above:$20/20$] (S2) at (0,1) {};
					
					\node[switch]  (W1) at (-3,0) {};
					\node[switch, dotted]  (W12) at (-1.5,0) {};
					\node[switch]  (W2) at (0,0) {};
					
					\node[puits]  (P1) at (-3,-1) {};
					\node[puits]  (P2) at (0,-1) {};
					\node[puits]  (P3) at (-1.5,1) {};
					
					\draw[->, >=latex] (S1) -- (W1) node [midway, right] {\small 50};
					\draw[->, >=latex] (S1) -- (P3) node [midway, above] {\small 10};
					\draw[->, >=latex] (S2) -- (W2) node [midway, left] {\small 20};
					
					\draw[->, >=latex] (W1) -- (P1) node [midway, right] {\small 50};
					\draw[<-, >=latex, dotted] (W12) -- (W1);
					\draw[<-, >=latex, dotted] (W12) -- (W2);
					\draw[->, >=latex] (W2) -- (P2) node [midway, left] {\small 20};
				\end{tikzpicture}
				\caption{$\minload{\mathcal{O}} = 0.6$, $\maxload{\mathcal{O}} = 1$, $\loadreserve{\mathcal{O}} = 0.4$.}
				\label{fig:example:b}
			\end{subfigure}\\
			\begin{subfigure}[t]{0.4\textwidth}
				\begin{tikzpicture}
					\clip (-3.75,-1.5) rectangle (0.75, 2);
					
					\node[source] [label=right:$$]  [label=above:$5/100$]  (S1) at (-3,1) {};
					\node[source] [label=left:$$] [label=above:$75/20$] (S2) at (0,1) {};
					
					\node[switch]  (W1) at (-3,0) {};
					\node[switch] [label=below:$55/20$] (W12) at (-1.5,0) {};
					\node[switch]  (W2) at (0,0) {};
					
					\node[puits]  (P1) at (-3,-1) {};
					\node[puits]  (P2) at (0,-1) {};
					\node[puits]  (P3) at (-1.5,1) {};
					
					\draw[<-, >=latex] (S1) -- (W1) node [midway, right] {\small 5};
					\draw[->, >=latex] (S1) -- (P3) node [midway, above] {\small 10};
					\draw[->, >=latex] (S2) -- (W2) node [midway, left] {\small 75};
					
					\draw[->, >=latex] (W1) -- (P1) node [midway, right] {\small 50};
					\draw[->, >=latex] (W12) -- (W1) node [midway, above] {\small 55};
					\draw[<-, >=latex] (W12) -- (W2) node [midway, above] {\small 55};
					\draw[->, >=latex] (W2) -- (P2) node [midway, left] {\small 20};
				\end{tikzpicture}
				\caption{The capacity constraint is not satisfied for $w_2$ and $s_2$.}
				\label{fig:example:c}
			\end{subfigure}%
			\begin{subfigure}[t]{0.4\textwidth}
			\begin{tikzpicture}
				\clip (-3.75,-1.5) rectangle (0.75, 2);
				
				\node[source] [label=above:$70/100$] (S1) at (-3,1) {};
				\node[source] [label=above:$10/20$] (S2) at (0,1) {};
				
				\node[switch]  (W1) at (-3,0) {};
				\node[switch]  (W12) at (-1.5,0) {};
				\node[switch]  (W2) at (0,0) {};
				
				\node[puits]  (P1) at (-3,-1) {};
				\node[puits]  (P2) at (0,-1) {};
				\node[puits]  (P3) at (-1.5,1) {};
				
				\draw[->, >=latex] (S1) -- (W1) node [midway, right] {\small 60};
				\draw[->, >=latex] (S1) -- (P3) node [midway, above] {\small 10};
				\draw[->, >=latex] (S2) -- (W2) node [midway, left] {\small 10};
				
				\draw[->, >=latex] (W1) -- (P1) node [midway, right] {\small 50};
				\draw[<-, >=latex] (W12) -- (W1) node [midway, above] {\small 10};
				\draw[->, >=latex] (W12) -- (W2) node [midway, above] {\small 10};
				\draw[->, >=latex] (W2) -- (P2) node [midway, left] {\small 20};
			\end{tikzpicture}
		\caption{$\minload{\mathcal{O}} = 0.5$, $\maxload{\mathcal{O}} = 0.70$, $\loadreserve{\mathcal{O}} = 0.2$.}
			\label{fig:example:d}
		\end{subfigure}\\
		\begin{subfigure}[t]{0.4\textwidth}
			\begin{tikzpicture}
				\clip (-3.75,-1.5) rectangle (0.75, 2);
				
				\node[source] [label=right:$$] (S1) at (-3,1) {};
				\node[source] [label=left:$$] (S2) at (0,1) {};
				
				\node[switch]  (W1) at (-3,0) {};
				\node[switch] (W12) at (-1.5,0) {};
				\node[switch]  (W2) at (0,0) {};
				
				\node[puits]  (P1) at (-3,-1) {};
				\node[puits]  (P2) at (0,-1) {};
				\node[puits]  (P3) at (-1.5,1) {};
				
				\draw[->, >=latex] (S1) -- (W1);
				\draw[->, >=latex] (S1) -- (P3);
				\draw[->, >=latex] (S2) -- (W2);
				
				\draw[->, >=latex] (W1) -- (P1);
				\draw[->, >=latex] (W12) -- (W1);
				\draw[<-, >=latex] (W12) -- (W2);
				\draw[<-, >=latex] (W2) -- (P2);
			\end{tikzpicture}
			\caption{The demand constraint is not satisfied as $p_2$ has no entering arc.}
			\label{fig:example:e}
		\end{subfigure}%
		\begin{subfigure}[t]{0.4\textwidth}
			\begin{tikzpicture}
				\clip (-3.75,-1.5) rectangle (0.75, 2);
				
				\node[source] (S1) at (-3,1) {};
				\node[source] (S2) at (0,1) {};
				
				\node[switch]  (W1) at (-3,0) {};
				\node[switch]  (W12) at (-1.5,0) {};
				\node[switch]  (W2) at (0,0) {};
				
				\node[puits]  (P1) at (-3,-1) {};
				\node[puits]  (P2) at (0,-1) {};
				\node[puits]  (P3) at (-1.5,1) {};
				
				\draw[->, >=latex] (S1) -- (W1) node [midway, right] {\small 25};
				\draw[->, >=latex] (S1) -- (P3) node [midway, above] {\small 10};
				\draw[->, >=latex] (S2) -- (W2) node [midway, left] {\small 10};
				
				\draw[->, >=latex] (W1) -- (P1) node [midway, right] {\small 50};
				\draw[->, >=latex] (W12) -- (W1) node [midway, above] {\small 25};
				\draw[->, >=latex] (W12) -- (W2) node [midway, above] {\small 10};
				\draw[->, >=latex] (W2) -- (P2) node [midway, left] {\small 20};
			\end{tikzpicture}
			\caption{The demand constraint is not satisfied as $w_2$ has no entering arc.}
			\label{fig:example:f}
		\end{subfigure}\\
			\caption{Example of an instance and orientations for that instance. Throughout the paper, we use triangles, circles and squares to represent sources, switches and sinks respectively. For each orientation, it is explained if it is feasible according to Definition~\ref{def:valid:3}. In this case, the metrics introduced in Definition~\ref{def:valid:6} are given.}
			\label{fig:example}
		\end{figure}
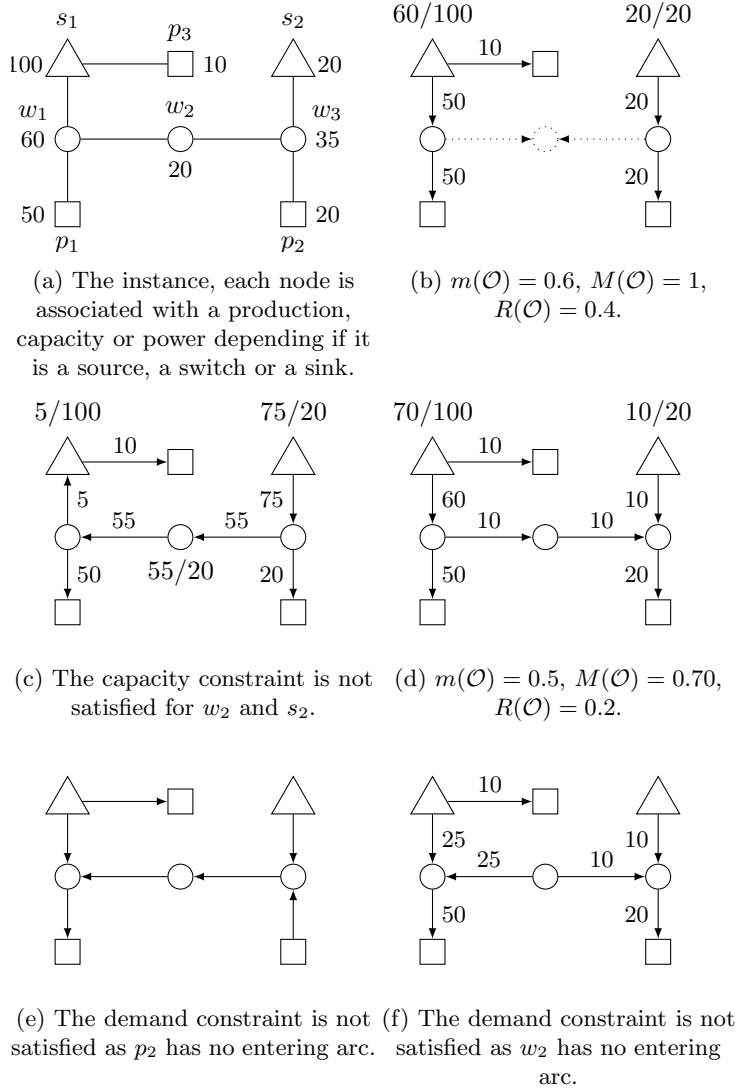
	\end{example}
	
	Determining whether there is a feasible orientation for a given electrical network leads to the following decision problem, which is shown to be NP-complete for general graphs \cite{Barth2021}, but polynomial for trees \cite{Barth2022}.
	
	\begin{problem}[\valid]
		Given a graph $G=(V,E)$, three functions  \textit{\productionfunctionname}, \textit{\capacityfunctionname} and \textit{\powerfunctionname}, does there exist a feasible orientation $\mathcal{O}$ of $G$.
	\end{problem}

	\begin{theorem}{\cite{Barth2022}}
		\label{theo:valid:poly}
		If $G$ is a tree, \valid can be solved in time $O(n^4 \log(n))$. In case of a positive answer, the algorithm builds a feasible orientation.
	\end{theorem}

	In the rest of the paper, we focus on this particular case where the input graph is a tree $T$. 

	\subsection{Optimization version}

	To evaluate the quality of a solution $\mathcal{O}$, we consider the loads of the sources.
	
	\begin{definition}
		\label{def:valid:6}
		Given a feasible orientation $\mathcal{O}$, the load of a source $s \in S$ is defined as the ratio $\load{\mathcal{O}}{s} = \flow{\mathcal{O}}{s} / \production{s}$.
	\end{definition}
	
	Maximizing the minimum load, minimizing the maximum load or minimizing the load reserve leads to a more balanced solicitation rate of the sources, knowing that the more balanced these rates are, the more resistant the network is to a reconfiguration in case of failures \cite{Guo2016,Li2014}. Reliability and resilience of the network are thus related to the following three optimization problems.
	
	\begin{problem}[\maxminm]
		Given a tree $T=(V,E)$, three functions  \textit{\productionfunctionname}, \textit{\capacityfunctionname} and \textit{\powerfunctionname}, compute an orientation $\mathcal{O}$ of $T$ satisfying the demand and the capacity constraints and maximizing the \emph{minimum load} $\minload{\mathcal{O}} = \min\limits_{s \in S} \load{\mathcal{O}}{s}$.
	\end{problem}

	\begin{problem}[\minmaxm]
		Given a tree $T=(V,E)$, three functions  \textit{\productionfunctionname}, \textit{\capacityfunctionname} and \textit{\powerfunctionname}, compute an orientation $\mathcal{O}$ of $T$ satisfying the demand and the capacity constraints and minimizing the \emph{maximum load} $\maxload{\mathcal{O}} = \max\limits_{s \in S} \load{\mathcal{O}}{s}$.
	\end{problem}

	\begin{problem}[\minr]
		Given a tree $T=(V,E)$, three functions  \textit{\productionfunctionname}, \textit{\capacityfunctionname} and \textit{\powerfunctionname}, compute an orientation $\mathcal{O}$ of $T$ satisfying the demand and the capacity constraints and minimizing the \emph{load reserve} $\loadreserve{\mathcal{O}} = \maxload{\mathcal{O}} - \minload{\mathcal{O}}$
	\end{problem}

	\begin{example}
		Figure~\ref{fig:example} gives examples of values for $\minload{\mathcal{O}}$, $\maxload{\mathcal{O}}$ and $\loadreserve{\mathcal{O}}$. Note that this example shows that maximizing $\minload{\mathcal{O}}$ is not equivalent to minimize $\loadreserve{\mathcal{O}}$. Similarly, we could build an example where the solutions minimizing $\maxload{\mathcal{O}}$ are not the solutions minimizing $\loadreserve{\mathcal{O}}$. 
	\end{example}
	
\section{Complexity of \maxminm, \minmaxm, \minr}

In this section, we explore the classical complexity of the three problems. We show that \minmaxm is the easiest problem, since it is polynomial, while the other two are strongly NP-Hard. 

\begin{theorem}
	\label{theo:minmaxm:poly}
	\minmaxm is polynomial
\end{theorem}
\begin{proof}
	The core idea to solve \minmaxm is to use a dichotomy with the algorithm of Theorem~\ref{theo:valid:poly}. The dichotomy works as follows: at each iteration, we manipulate an interval $[a, b]$ containing the optimal maximum load (initialized with $a = 0$ and $b = 1$). We search for an orientation $\mathcal{O}$ with $\maxload{\mathcal{O}} \leq (a + b) / 2$. If no such orientation exists, then $a$ becomes $(a + b) / 2$, otherwise, $b$ becomes $\maxload{\mathcal{O}}$. 
	
	To do so, we solve the construction problem associated to \minmaxm that, given a rational $M \in [0; 1]$, asks for a feasible orientation $\mathcal{O}$ where $\maxload{\mathcal{O}} \leq M$. This problem can be answered by determining the validity of the instance where, for every source $s \in S$, the production $\production{s}$ is truncated to $M \cdot \production{s}$. By the capacity constraint, the flow outgoing from $s$ is less than $M \cdot \production{s}$, consequently, any feasible orientation $\mathcal{O}$ in that truncated instance is also feasible in the original instance and satisfies $\maxload{\mathcal{O}} \leq M$. Note that each production should be an integer and that $M \cdot \production{s}$ can be not in $\mathbb{N}$. We can get around this problem by noticing that multiplying each production, capacity and power by the same value does not change the instance. If $M = p / q$, then we can then multiply all the values by $q$ to get an equivalent instance where each production, power and capacity is an integer.
	
	The load reserve of any feasible orientation $\mathcal{O}$ is a rational that can be encoded with a polynomial number of bits in the size of the instance. Then there exists a value $\varepsilon$ such that $\lfloor \log(\varepsilon) \rfloor$ is polynomial and, for any two distinct feasible orientations $\mathcal{O}$ and $\mathcal{O}'$, $\lvert \maxload{\mathcal{O}} - \maxload{\mathcal{O}'}\rvert  \geq \varepsilon$. So when $b - a < \varepsilon$, after $O(\log(\varepsilon))$ iterations, we stop and return the last found orientation $\mathcal{O}$ and the corresponding value $\maxload{\mathcal{O}}$.
	
	Let $\sigma = \max_{s \in S} \production{s}$ and $\Pi = \sum_{p \in P} \power{p}$. The load of a source is either 0 or belongs to $L = \{p / q | p \in \left[1; \Pi\right], q \in [1, n^n \cdot \sigma]\}$ Indeed, the flow cannot be greater that the value asked by the consumers, it cannot be divided more than $n$ times and is divided by at most $n$ at each step, finally, in order to compute the load, we divide the flow by at most the maximum source production. The value $\varepsilon$ is at least the minimum difference between two load reserves. Each load reserve is a difference between two values of $L$. Thus $\varepsilon \geq 1/(n^n \cdot \sigma)^4$. The number of iterations is then at most $4 \cdot n \log(n \sigma)$. By Theorem~\ref{theo:valid:poly}, the dichotomy is done in time at most $O(n^5 \log(n) \log(n \sigma))$.
\end{proof}

\begin{theorem}
	\label{theo:maxminm:minr:nphard}
	\maxminm and \minr are strongly NP-hard. More precisely determining if there exists a feasible orientation $\mathcal{O}$ with $\loadreserve{\mathcal{O}} \leq \frac{2}{3}$ (respectively $\minload{\mathcal{O}} \geq \frac{1}{3}$) is NP-Complete. 
\end{theorem}
\begin{proof}
	We describe below a reduction from the Subset Sum problem where, given $n$ integers $x_1, x_2, \dots, x_n$ and an integer $B$, the question is to decide whether there exists $I \subseteq \llbracket 1; n \rrbracket$ such that $\sum_{i \in I} x_i = B$. Let then $\mathcal{I} = (x_1, x_2, \dots, x_n, B)$ be such an instance.
	
	This proof is done in three steps: first we describe a gadget to encode large integers, then we describe the construction of the reduction and some properties, and finally we prove the correctness of the reduction.
	
	In the built instance, we will have $\maxload{\mathcal{O}} = 1$ if $\mathcal{O}$ is feasible, meaning that maximizing the minimum load and minimizing the load reserve are equivalent. 
	
	\paragraph{Encode large numbers}
	
	As Subset Sum is weakly NP-Complete, we cannot use the values $x_i$ or $B$ in the powers, capacities and productions of the instance. First, we explain how it is possible to encode these integers with small productions and powers. Let $x \in \mathbb{N}$ be any integer such that a binary encoding of $x + 2$ is $\overline{b_{p - 1}\cdots b_1b_0}$ ($b_0$ is the lowest bit and $b_{p-1}$ is the highest). Let $m \geq p$ such that $m \geq 3$, we set $b_{p} = b_{p + 1} = \dots = b_{m - 1} = 0$.  Then, we have $x + 2 = \overline{b_{m-1}\dots b_1b_0} = \sum_{j = 0}^{m - 1} b_j 2^{j}$. 
	
	We then build a path of size $2m$ containing alternatively, for $j \in \llbracket 1; m \rrbracket$, a sink $p_j \in P$ with power $\power{p_j} = 2 + b_{j-1}$ and a source $s_j \in S$ with production $\production{s_j} = 2$ if $j = 1$ and 3 otherwise. Finally, we assume that the source $s_m$ has another neighbor $v$ (which can be of any type). Figure~\ref{fig:encode:2} illustrates this instance. 
	
	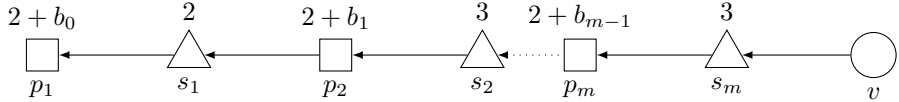
\begin{figure}[ht!]
		\centering
		\begin{tikzpicture}[xscale=0.65]
			\tikzset{source/.style={draw, regular polygon, regular polygon sides = 3, minimum height=0.6cm}}	
			\tikzset{switch/.style={draw, circle, minimum height=0.6cm}}
			\tikzset{puits/.style={draw, regular polygon, regular polygon sides = 4, minimum height=0.6cm}}
			
			\node[puits, label={above:$2 + b_0$}, label={below:$p_1$}] (P1) at (0, 0) {};
			\node[puits, label={above:$2 + b_1$}, label={below:$p_2$}] (P2) at (6, 0) {};
			\node[puits, label={above:$2 + b_{m - 1}$}, label={below:$p_m$}] (PM) at (11, 0) {};
			
			\node[source, label={above:$2$}, label={below:$s_1$}] (S1) at (3, 0) {};
			\node[source, label={above:$3$}, label={below:$s_2$}] (S2) at (9, 0) {};
			\node[source, label={above:$3$}, label={below:$s_m$}] (SM) at (14, 0) {};
			
			\node[switch, label={below:$v$}] (V) at (17, 0) {};
			
			\draw[<-, >=latex] (P1) -- (S1);
			\draw[<-, >=latex] (P2) -- (S2);
			\draw[<-, >=latex] (PM) -- (SM);
			
			\draw[<-, >=latex] (S1) -- (P2);
			\draw[<-, >=latex, dotted] (S2) -- (PM) node[midway, sloped, below] {};
%			\draw[<-] (SM) -- ($(SM) + (1, 0)$) node[right] {\small $2 + x \cdot 2^{-m}$};
			\draw[<-, >=latex] (SM) -- (V);
			
		\end{tikzpicture}
		\caption{Gadget equivalent to a sink with power $2 + x \cdot 2^{-m}$ where $x + 2 = \overline{b_{m - 1}\cdots b_1b_0}$. On the graph, each edge is directed toward the left as, otherwise, some demand constraints or some capacity constraints are not satisfied. Below each arc is given the flow going through the arc. }	
		\label{fig:encode:2}
	\end{figure}
	
	We now show an intermediate result showing that this path acts as a sink of value $2 + x \cdot 2^{-m}$.
	
	\begin{fact}
		\label{lem:gadget}
		Given a feasible orientation $\mathcal{O}$, then, for every $j \in \llbracket 1; m \rrbracket$, $\mathcal{O}([p_j, s_j]) = (s_j, p_j)$, for every $j \in \llbracket 1; m - 1 \rrbracket$, $\mathcal{O}([p_{j+1}, s_j]) = (p_{j+1}, s_j)$ and $\mathcal{O}([v, s_m]) = (v, s_m)$. The load of each source $s_j$ is at least 0.5 and $\flow{\mathcal{O}}{v, s_m} = 2 + x \cdot 2^{-m}$.
	\end{fact}
	\begin{proof}
		We prove this property by induction on $j$, proving in addition that $\flow{\mathcal{O}}{s_j, p_j} = \frac{2^{j+1} - 2}{2^{j - 1}} + \sum_{i = 0}^{j - 1} b_i 2^{i + 1 - j}$. 
		
		$[s_1, p_1]$ is directed toward $p_1$ otherwise the demand constraint is not satisfied for $p_1$. In addition, $\mathcal{O}([s_1, p_2]) = (p_2, s_1)$ otherwise the flow outgoing from $s_1$ is at least $3$ which is greater than $\production{s_1}$. The flow entering $p_1$ is then $2 + b_0 = \frac{2^{2} - 2}{2^{0}} + \sum_{i = 0}^{0} b_i 2^{i}$. Thus, the property is satisfied for $j = 1$. 
		
		Assuming the property is true for all $i \leq j$, we now show it is also true for $j + 1$. The edge $[s_{j+1}, p_{j+1}]$ is directed toward $p_{j+1}$ otherwise the demand constraint is not satisfied for $p_{j+1}$. The flow $\flow{\mathcal{O}}{s_{j+1}, p_{j+1}}$ entering $p_{j + 1}$ is then
		\begin{align*}
			2 + b_j + \frac{1}{2} \left(\frac{2^{j+1} - 2}{2^{j - 1}} + \sum\limits_{i = 0}^{j - 1} b_i 2^{i + 1 - j}\right) &= \left(2 + \frac{2^{j+1} - 2}{2^j}\right) + \left(\sum\limits_{i = 0}^{j-1} b_i 2^{i - j} + b_j\right)\\
			&= \frac{2^{j+2} - 2}{2^j} + \sum\limits_{i = 0}^{j} b_i 2^{i - j}
		\end{align*}
		
		If $j + 1 \neq m$ then $\mathcal{O}([s_{j+1}, p_{j+2}]) = (p_{j+2}, s_{j+1})$ otherwise $s_{j+1}$ must power at least the flow $\flow{\mathcal{O}}{s_{j+1}, p_{j+1}}$ and half of $\power{p_{j+2}}$, which is strictly greater than $(2^{j+2} - 2)/2^j \geq 3 = \production{s_{j+1}}$. If $j + 1 = m$ then $s_m$ cannot power $p_m$ alone as $j + 1 = m \geq 3$ and then $(2^{j+2} - 2)/2^j > 3 = \production{s_m}$. The property is then proved by induction. 
		
		Consequently, the flow $\flow{\mathcal{O}}{v, s_m}$ entering $s_m$ is the flow entering $p_m$ divided by 2. 
		$$\flow{\mathcal{O}}{v, s_m} = \frac{2^{m+1} - 2}{2^m} + \sum\limits_{i = 0}^{m - 1} b_i 2^{i - m} = 2 + (\sum\limits_{i = 0}^{m - 1} b_i 2^{i} - 2) \cdot 2^{-m} = 2 + x \cdot 2^{-m}$$
		
		Note that $\flow{\mathcal{O}}{s_j} = \flow{\mathcal{O}}{s_j, p_j} / 2$ is between 1.5 and 2.5 except when $j = 1$ where it is between 1 and 1.5. Then $\load{\mathcal{O}}{s_j}$ is greater than 0.5.
	\end{proof}

	Thus we can now use sinks with called powers $2 - 2^{-m} x$ for any integer $x$.

	\paragraph{Build the instance}
	
	In this part, we will explain how we build an instance $\mathcal{J} = (T, \productionfunctionname, \capacityfunctionname, \powerfunctionname)$ of \minr and \maxminm from $\mathcal{I} = (x_1, x_2, \dots, x_n, B)$ using the previous gadget. 
	
	Without loss of generality, we assume that $n$ is such that $0.1 > \frac{2}{3n + 1}$ and that $B \geq x_i$ for all $i \in \llbracket 1; n \rrbracket$. We then choose $m$ such that 
	
	\begin{align*}
		1 > 2^{-m} B &; & -1 < 2^{-m} (B - \sum\limits_{i = 1}^n x_i) &; & \forall i \in \llbracket 1; n \rrbracket \quad \frac{2}{3n + 1} > 2^{-m} 2 x_i
	\end{align*}

	Note that $m$ is polynomial in $n$ and the size of the encoding of $x_i$ and $B$.
	\begin{itemize}
		\item We add $2n + 2$ switches $w_1, w_2, \dots, w_n, v_1, v_2, \dots, v_n, w$ and $w_c$ with arbitrarily large capacities (for instance $\sum_{p \in P} \power{p}$).
		\item We build $n$ sinks $p_i$ with power $4 - 2^{-m} 2 x_i$, for each $i \in \llbracket 1; n \rrbracket$, by building the gadget of Figure~\ref{fig:encode:2} with $x = 2^m - 2x_i$ and where the node $v$ is replaced by a sink with power 1. We also build a sink $p$ with power $2 + 2^{-m}B$ by using the gadget with $x = B$ where the node $v$ is the switch $w$. Note that this procedure adds $(2m + 1) \cdot n + 2m$ nodes.
		\item We add $n$ sinks $q_1, q_2, \dots, q_n$ with power 4 and a sink $p_c$ with power 10. 
		\item We add $2n$ sources $t_1, t_2, \dots, t_n, r_1, r_2, \dots, r_n$ with production 4; $n$ sources $s_1, s_2, \dots, s_n$ with production 2; a source $s$ with production $6$ and a source $s_c$ with production 10.
		\item For each $i \in \llbracket 1; n \rrbracket$, we add the 7 edges $[p_i, t_i]$, $[t_i, w_i]$, $[w_i, w]$, $[q_i, r_i]$, $[r_i, v_i]$, $[v_i, w]$ and $[s_i, w]$. 
		\item We add the edges $[s_n, w_c], [w_c, s_c]$, $[p_c, s_c]$ and $[s, w]$.
	\end{itemize}
	Figure~\ref{fig:reduc:2} illustrates this instance. 
	
	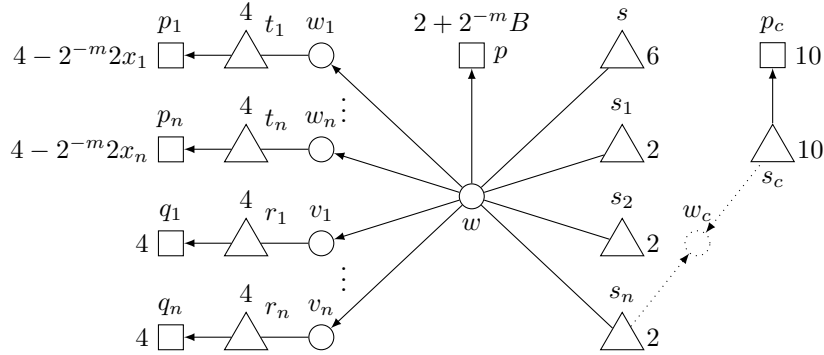
\begin{figure}[ht!]
		\centering
			\begin{tikzpicture}[yscale=1.25]

				\node[puits, label={left:$4 - 2^{-m} 2x_1$}, label={above:$p_1$}] (P1) at (0, 6) {};
				\node[puits, label={left:$4 - 2^{-m} 2x_n$}, label={above:$p_n$}] (PN) at (0, 5) {};
				\node[puits, label={left:$4$}, label={above:$q_1$}] (Q1) at (0, 4) {};
				\node[puits, label={left:$4$}, label={above:$q_n$}] (QN) at (0, 3) {};
				
				\node[source, label={above:$4$}, label={above right:$t_1$}] (T1) at (1, 6) {};
				\node[source, label={above:$4$}, label={above right:$t_n$}] (TN) at (1, 5) {};
				\node[source, label={above:$4$}, label={above right:$r_1$}] (R1) at (1, 4) {};
				\node[source, label={above:$4$}, label={above right:$r_n$}] (RN) at (1, 3) {};
				
				\node[switch, label={above:$w_1$}] (W1) at (2, 6) {};
				\node[switch, label={above:$w_n$}] (WN) at (2, 5) {};
				\node[switch, label={above:$v_1$}] (V1) at (2, 4) {};
				\node[switch, label={above:$v_n$}] (VN) at (2, 3) {};
				
				\node[switch, label={below:$w$}] (W) at (4, 4.5) {};
				\node[puits, label={above:$2 + 2^{-m} B$},, label={right:$p$}] (P) at (4, 6) {};
				
				\node[source, label={right:$6$}, label={above:$s$}] (S) at (6, 6) {};
				\node[source, label={right:$2$}, label={above:$s_1$}] (S1) at (6, 5) {};
				\node[source, label={right:$2$}, label={above:$s_2$}] (S2) at (6, 4) {};
				\node[source, label={right:$2$}, label={above:$s_n$}] (SN) at (6, 3) {};

				\node[switch, label={above:$w_c$}, dotted] (A) at ($(SN) + (1, 1)$) {};
				\node[source, label={below:$s_c$}, label={right:10}] (C) at ($(A) + (1, 1)$) {};
				\node[puits, label={above:$p_c$}, label={right:10}] (B) at ($(C) + (0, 1)$) {};

				\draw ($(W1)!0.5!(WN) + (0.3, 0)$) node {$\vdots$};
				\draw ($(V1)!0.5!(VN) + (0.3, 0.2)$) node {$\vdots$};
				
				\draw[<-, >=latex] (P1) -- (T1);
				\draw[<-, >=latex] (PN) -- (TN);
				
				\draw[<-, >=latex] (Q1) -- (R1);
				\draw[<-, >=latex] (QN) -- (RN);
				
				\draw[-, >=latex] (T1) -- (W1);
				\draw[-, >=latex] (TN) -- (WN);
				
				\draw[-, >=latex] (R1) -- (V1);
				\draw[-, >=latex] (RN) -- (VN);
				
				\draw[<-, >=latex] (W1) -- (W);
				\draw[<-, >=latex] (WN) -- (W);
				\draw[<-, >=latex] (V1) -- (W);
				\draw[<-, >=latex] (VN) -- (W);
				\draw[<-, >=latex] (P) -- (W);
				\draw[-, >=latex] (S) -- (W);
				\draw[-, >=latex] (S1) -- (W);
				\draw[-, >=latex] (S2) -- (W);
				\draw[-, >=latex] (SN) -- (W);
				
				\draw[->, >=latex, dotted] (SN) -- (A);
				\draw[<-, >=latex, dotted] (A) -- (C);
				\draw[<-, >=latex] (B) -- (C);
				
			\end{tikzpicture}
			\caption{Instance $\mathcal{J}$ of the reduction. Some edges are directed on the drawing, other orientations for these edges imply that the capacity or the demand constraints are not satisfied. The powers and productions are written next to each sink and source. Each switch can be considered to have an infinite capacity. Note that the switch $w_c$ is not activated as no arc goes out of this switch.}	
			\label{fig:reduc:2}
	\end{figure}

	This tree contains $N = 2m \cdot (n + 1) + 7n + 5$ nodes (recall that $p$ and each sink $p_i$ are gadgets creating $(2m + 1) \cdot n + 2m$ additional nodes), every power and production is between 2 and 10 and every capacity is at most $\sum_{p \in P} \power{p} \leq 10 \cdot N$. This construction of $\mathcal{J}$ is then done in polynomial time.

	\paragraph{Properties of a feasible orientation of $\mathcal{J}$}	 
	
	Let $\mathcal{O}$ be any feasible orientation of $\mathcal{J}$. Due to the demand constraint, $\mathcal{O}([s_c, p_c])  = (s_c, p_c)$. Then the switch $w_c$ cannot be activated (the two incident edges of $w_c$ are directed toward it), otherwise either $s_c$ powers some sink through $w_c$ which exceeds its production, or $s_n$ powers at least a fraction $\frac{1}{4}$ of $\power{p_c}$, that is 2.5, which exceeds its production too. Consequently, $s_c$ powers $p_c$ alone and $\load{\mathcal{O}}{c} = 1$. 
	
	Note that if $\mathcal{O}$ directs one of the edges $[s, w]$ or $[s_i, w]$, for some $i \in \llbracket 1; n \rrbracket$, from $w$ then the flow outgoing from the source is 0. In that case, the load reserve $\loadreserve{\mathcal{O}}$ is 1 and the minimum load $\minload{\mathcal{O}}$ is 0. We then consider only feasible orientations where this case does not occur. 
	
	Again, due to the demand constraint, $\mathcal{O}([w, p])  = (w, p)$ and, for each $i \in \llbracket 1; n \rrbracket$, we have $\mathcal{O}([p_i, t_i])  = (t_i, p_i)$ and $\mathcal{O}([q_i, r_i])  = (r_i, q_i)$. 
	
	We now assume that $\mathcal{O}([w, w_i])  = (w_i, w)$ for some $i \in \llbracket 1; n \rrbracket$. The flow $\flow{\mathcal{O}}{w_i, w}$ is at least $\flow{\mathcal{O}}{w, p} / (3n + 1) \geq 2 / (3n + 1)$. We can now examine the orientation of $[t_i, w_i]$. If $\mathcal{O}([t_i, w_i]) = (w_i, t_i)$, the demand constraint is not satisfied for $w_i$ which has no entering arc. On the other hand, if $\mathcal{O}([t_i, w_i]) = (t_i, w_i)$, then $\flow{\mathcal{O}}{t_i} = \flow{\mathcal{O}}{t_i, p_i} + \flow{\mathcal{O}}{t_i, w_i} \geq 4 - 2^{-m} 2 x_i + 2 / (3n + 1)$. As we chose $m$ so that $2 / (3n + 1) > 2^{-m} 2 x_i$ and as the production of $t_i$ is 4, the capacity constraint of $t_i$ is not satisfied. Consequently, $\mathcal{O}([w, w_i])  = (w, w_i)$. Similarly, $\mathcal{O}([w, v_i])  = (w, v_i)$.
	
	Hence, the only degree of freedom we have is to decide the orientation of the edges $[t_i, w_i]$ and $[r_i, v_i]$.
	
	We now detail the load of each source: 
	\begin{itemize}
		\item $\load{\mathcal{O}}{s_c}= 1$
		\item For each $i \in \llbracket 1; n \rrbracket$, $\load{\mathcal{O}}{r_i} = 1$ if $\mathcal{O}([r_i, v_i]) = (r_i, v_i)$ and $0.5$ otherwise.
		\item For each $i \in \llbracket 1; n \rrbracket$, $\load{\mathcal{O}}{t_i} = 1 - 2^{-m} x_i / 2$ if $\mathcal{O}([t_i, w_i]) = (t_i, w_i)$ and $0.5 - 2^{-m} x_i / 4$ otherwise. Note that $m$ has been chosen so that $2^{-m}2x_i \leq 0.1$. Consequently, the load of $t_i$ is at least $0.4875$.
		\item By Fact~\ref{lem:gadget}, the load of each hidden source in the gadgets is at least $0.5$.
		\item Let $\alpha_v$ be the indexes $i \in \llbracket 1; n \rrbracket$ such that $\mathcal{O}([r_i, v_i]) = (v_i, r_i)$ and $\alpha_w$ be those for which $\mathcal{O}([t_i, w_i]) = (w_i, t_i)$. 
		\begin{align*}
			& \load{\mathcal{O}}{s} = \frac{1}{\production{s} \cdot (n + 1)}\left( 2 + 2^{-m}B + \sum\limits_{j \in \alpha_w} (2 - 2^{-m}x_j) + \sum\limits_{j \in \alpha_v} 2 \right) \\
			&= \frac{1}{6 \cdot (n + 1)}\left(2 \cdot (1 + \lvert \alpha_w\rvert  + \lvert \alpha_v\rvert ) + 2^{-m} (B - \sum\limits_{j \in \alpha_w}x_j)\right)
		\end{align*}
	\item For each $i \in \llbracket 1; n \rrbracket$, we have $\load{\mathcal{O}}{s_i} = 3 \cdot \load{\mathcal{O}}{s}$. 
	\end{itemize}

	 The load of $s$ cannot be greater than $\frac{1}{3}$ as, otherwise, for each $i \in \llbracket 1; n \rrbracket$, we would have $\load{\mathcal{O}}{s_i} = 3 \cdot \load{\mathcal{O}}{s} > 1$, which means that $\mathcal{O}$ is not feasible. Consequently, $\minload{\mathcal{O}} = \load{\mathcal{O}}{s}$. 
	
	\paragraph{Equivalence of the instances}

	We end this proof by demonstrating that the minimum load reserve of $\mathcal{J}$ is exactly $\frac{2}{3}$ if and only if $\mathcal{I}$ is positive. Otherwise it is strictly greater.
	
	If $\mathcal{I}$ is positive, then there exists $I \subseteq \llbracket 1; n \rrbracket$ such that $\sum_{i \in I} x_i = B$. Let $\mathcal{O}$ be the orientation where $\mathcal{O}([t_i, w_i]) = (w_i, t_i)$ and $\mathcal{O}([r_i, v_i]) = (r_i, v_i)$ if and only if $i \in I$. In that case $\alpha_w = I$, $\alpha_v = \llbracket 1; n \rrbracket \backslash I$ and $\sum_{j \in \alpha_w}x_j = B$. Thus $\load{\mathcal{O}}{s} = \frac{1}{6 \cdot (n + 1)}\left(2 \cdot (1 + n) + 0\right) = \frac{1}{3}$. The load of the sources $s_i$ is 1. Thus $\mathcal{O}$ is feasible, $\loadreserve{\mathcal{O}} = \frac{2}{3}$ and $\minload{\mathcal{O}} = \frac{1}{3}$.
	
	If $\mathcal{I}$ is negative, then whatever the feasible orientation is, we have the relation $2^{-m} (B - \sum_{i: \alpha(w_i) = 1}x_i) \neq 0$. In addition, this value is lower than $2^{-m} B < 1$ and greater than $2^{-m} (B - \sum_{i = 1}^n x_i) > -1$. \\Then $2 \cdot (1 + \lvert \alpha_w\rvert  + \lvert \alpha_v\rvert ) + 2^{-m} (B - \sum_{j \in \alpha_w}x_j)$ cannot be an integer, thus we have $\load{\mathcal{O}}{s} \neq \frac{1}{3}$. Consequently, $\load{\mathcal{O}}{s} < \frac{1}{3}$, $\loadreserve{\mathcal{O}} > \frac{2}{3}$ and $\minload{\mathcal{O}} < \frac{1}{3}$.
	
	This proves that \minr and \maxminm are strongly NP-Hard. 	
\end{proof}

\begin{remark}
	Theorem~\ref{theo:maxminm:minr:nphard} also proves that no interesting parameterized algorithm can be hoped for the decision versions of \minr and \maxminm (with respect to the natural parameters of those problems).
	
	Note also that we can arbitrarily increase the constant $\frac{2}{3}$ of Theorem~\ref{theo:maxminm:minr:nphard} (respectively decrease the value $\frac{1}{3}$) by increasing the production of $s$. We can also decrease this constant (respectively increase), but in order to keep $s$ to be the source minimizing the load, we need to use a trick used later in the proof of Theorem~\ref{theo:minr:noapprox} in order to arbitrarily increase the load of all the sources. 
\end{remark}

\section{Hardness of approximation for \minr}

This section is devoted to extending the reduction of Theorem~\ref{theo:maxminm:minr:nphard} in order to prove that there is no efficient approximation algorithm for \minr unless P = NP. Note that this proof does not apply to \maxminm. As we will see in the next section, there is indeed an approximation algorithm for this problem.

\begin{theorem}
	\label{theo:minr:noapprox}
	Unless P = NP, then, for any constant $c \in \mathbb{N}$, there is no polynomial approximation algorithm for \minr with ratio $2^{\lvert T\rvert ^c}$ for a given tree $T$.
\end{theorem}
\begin{proof}
	We use a construction that is similar to the one given in the proof of Theorem~\ref{theo:maxminm:minr:nphard} to build a GAP reduction from the Subset Sum problem. Let $\mathcal{J}$ be the instance obtained in that proof, described in Figure~\ref{fig:reduc:2}. We then use the same notations in this proof (particularly $n$, the number of variables $x_i$ in the Subset Sum instance). 
	
	Let $c \in \mathbb{N}$ and $r = 2^{\lvert T\rvert ^c}$.	We set three values $H \in \mathbb{N}$, $\xi, L \in \mathbb{Q}$ such that:
	\begin{align*}
		L = 2 + \frac{2^{-m}}{r \cdot (n + 1)} &;  & \frac{2}{L} < 1 - \xi &; & \frac{H}{H + 10} \geq 1 - \xi
	\end{align*}
	
	Those values are polynomial with respect to the size of $\mathcal{I}$ (recall that $\log(r)$ depends on $\lvert T\rvert $ which is polynomial in $\lvert \mathcal{I}\rvert $). 
	
	From $\mathcal{J}$, we build an instance $\mathcal{K}$ by increasing the load of all the sources except $s$. For each such source $s'$, we attach to $s'$ a sink with power $H$, and we increase $\production{s'}$ by $H$. Thus the load of $s'$ is at least $H / (H + 10) \geq 1 - \xi$. Note that this operation does not change the feasibility of any orientation. It also does not change the load of the source $s_c$ which is 1. For the source $s$, we replace $\production{s}$ by $L$. Note that this production is not an integer, but if we write $L = p / q$, we can simply multiply each power, capacity and production by $q$ to get an equivalent instance with integer productions, powers and capacities.
	
	Given a feasible orientation $\mathcal{O}$, the load of $s$ is at most $\frac{2}{L} < 1 - \xi$ thus $\min(\mathcal{O}) = \load{\mathcal{O}}{s}$. 
	
	Let $\alpha_v$ be the indexes $i \in \llbracket 1; n \rrbracket$ such that $\mathcal{O}([r_i, v_i]) = (v_i, r_i)$ and $\alpha_w$ be those for which $\mathcal{O}([t_i, w_i]) = (w_i, t_i)$. 
	\begin{align*}
		\load{\mathcal{O}}{s}= \frac{1}{L \cdot (n + 1)}\left(2 \cdot (1 + \lvert \alpha_w\rvert  + \lvert \alpha_v\rvert ) + 2^{-m} (B - \sum\limits_{j \in \alpha_w}x_j)\right)
	\end{align*}
	
	If $\mathcal{I}$ is positive, then there exists $I \subseteq \llbracket 1; n \rrbracket$ such that $\sum_{i \in I} x_i = B$. Let $\mathcal{O}$ be the orientation where $\mathcal{O}([t_i, w_i]) = (w_i, t_i)$ and $\mathcal{O}([r_i, v_i]) = (r_i, v_i)$ if and only if $i \in I$. In that case $\alpha_w = I$, $\alpha_v = \llbracket 1; n \rrbracket \backslash I$ and $\sum_{j \in \alpha_w}x_j = B$. Thus the load of $s$ is $\load{\mathcal{O}}{s} = \frac{1}{L \cdot (n + 1)}\left(2 \cdot (1 + n) + 0\right) = 2 / L$. The load of the sources $s_i$ is 1. Thus $\mathcal{O}$ is feasible, and $\loadreserve{\mathcal{O}} = 1 - 2 / L$.
	
	If $\mathcal{I}$ is negative, then whatever the feasible orientation is, we have the relation $2^{-m} (B - \sum_{i: \alpha(w_i) = 1}x_i) \neq 0$. The difference $B - \sum_{i: \alpha(w_i) = 1}x_i$ is necessarily negative (otherwise the capacity of $s$ is exceeded) and is an integer. Thus
	
	\begin{align*}
		2 \cdot (1 + \lvert \alpha_w\rvert  + \lvert \alpha_v\rvert ) + 2^{-m} (B - \sum\limits_{j \in \alpha_w}x_j) &\leq 2 \cdot (n + 1) - 2^{-m}\\
		\load{\mathcal{O}}{s} &\leq \frac{2}{L} - \frac{2^{-m}}{L \cdot (n + 1)}\\
		\loadreserve{\mathcal{O}} &\geq 1 - \frac{2}{L} + \frac{2^{-m}}{L \cdot (n + 1)}
	\end{align*}

	$$\text{However }
	\frac{1 - \frac{2}{L} + \frac{2^{-m}}{L \cdot (n + 1)}}{1 - \frac{2}{L}} = \frac{L - 2 + \frac{2^{-m}}{n + 1}}{L - 2} = \frac{\frac{2^{-m}}{r \cdot (n + 1)} + \frac{2^{-m}}{n + 1}}{\frac{2^{-m}}{r \cdot (n + 1)}} = 1 + r$$
	
	Consequently, any polynomial approximation for \minr with ratio strictly less than $1 + r$ could decide in polynomial time whether $\mathcal{I}$ is positive or not. Unless P = NP, there is no $r$-polynomial approximation for \minr. 
\end{proof}

\section{Approximability of \maxminm and \minr}

In this section, we prove that there exists an FPTAS\footnote{A Fully Polynomial Time Approximation Scheme (or FPTAS) for a minimization problem (respectively maximization problem) is an algorithm that, given a value $\varepsilon > 0$, returns a feasible solution for which the objective value is equal to the optimal value multiplied by at most $1 + \varepsilon$ (respectively $1 - \varepsilon$). This algorithm is polynomial in the size of the instance and in $1/\varepsilon$.} for \maxminm and an FPTAS with an absolute ratio\footnote{An FPTAS with absolute ratio for a minimization problem (respectively maximization problem) is an algorithm that, given a value $\varepsilon > 0$, returns a feasible solution for which the objective value is equal to the optimal value shifted by at most $\varepsilon$ (respectively $-\varepsilon$). This algorithm is polynomial in the size of the instance and in $1/\varepsilon$.} for \maxminm and \minr. Note that this last result does not contradict the inapproximability result of \minr as we do not use the same metric: there is no FPTAS (with a relative ratio) for \minr, but there exists an FPTAS with an absolute ratio. 

The key idea to overcome the hardness of the reduction of Theorem~\ref{theo:maxminm:minr:nphard} consists in noticing that, for many feasible solutions $\mathcal{O}$, the value $\minload{\mathcal{O}}$ is close to $1/3$. In particular, it seems that the flow obtained in the arc $(s, w)$ does not depend so much on the chosen orientation as this flow vary from $2 - \varepsilon$ to $2$ (where $\varepsilon$ is a very small value). The problem is NP-hard because it is hard to achieve the desired value $2$, but it is not hard to get a value really close to 2.  For example, one can set $\mathcal{O}([w_i, t_i]) = (w_i, t_i)$ and $\mathcal{O}([v_i, r_i]) = (r_i, v_i)$ for all $i$, to get $\load{\mathcal{O}}{s} = 1/3 + 2^{-m} \cdot (B - \sum\limits_{j = 1}^n x_j) / (6 \cdot (n + 1)) > 1/3 - 1 / (6 \cdot (n + 1))$.

Similarly, Theorem~\ref{theo:minr:noapprox} shows that it is hard to find a feasible solution with a small relative error for \minr. In the reduction, the load reserves of all the feasible orientations are close to 0 as all the loads are close to 1. It is then not hard to find a feasible solution with a very small load reserve, in which case we achieve a very small absolute error. However, the relative error is very high as the optimal value is almost 0.

In this context, one way to approximate the problem is to round the flow at each step of the calculation to get rid of the small variations. This defines, for each orientation, a new metric based on the rounded flow instead of the flow itself in which we focus on the rounded load of each source (that is the output rounded flow divided by the production of the source). Instead of searching for an optimal solution, we will focus on finding a feasible orientation that minimizes the rounded load reserve or maximizes the minimum rounded load. 

Three main questions arise from this strategy:
\begin{itemize}
	\item First, what makes the computation with the rounded flow easier? An important property of that flow is that there should be a polynomial number of possible values of a rounded flow. This allows to enumerate all these values and use them in a dynamic programming algorithm. Furthermore, this allows to enumerate all the possible values of rounded loads and thus all the possible objective values in polynomial time. To achieve such a result, it is necessary not to use a linear rounding, that is rounding to the closest fraction of some value $\varepsilon$. Indeed, the real flow belongs to an interval that has an exponential size. To overcome this difficulty, we use a rounding that depends on the value of the flow: the larger the flow is, the larger the rounding error is. To do so, we cut each interval $[2^i; 2^{i+1}]$ into $\left\lceil\frac{1}{\varepsilon}\right\rceil$ pieces and we round the flow to the closest piece. 
	\item Second, does this rounded flow provide a good approximation ratio? As we compute the rounding flow from the sinks to the sources, the rounding error propagates and increases. We show that with a polynomially small value $\varepsilon$ this rounding error can be bounded by an arbitrarily small constant.
	\item Third, our technique is, in some way, adapted from the FPTAS algorithm for the Knapsack problem \cite{Ibarra1975}. In this algorithm, the coefficients of the objective functions are rounded but not those of the constraints, which trivially makes the set of feasible solutions unchanged during the rounding process. In our problem, the objective function and the capacity constraints are based on the flow. Since the rounded flow is lower than the real flow, there may be some orientations that are feasible with respect to the rounded flow and not with the real flow. This is the tricky part of the algorithm. We want the rounded flow to be used only to define a new metric that is easier to optimize while not changing the set of feasible solutions. Thus, we must somehow keep the real flow in the computation to ensure that the capacity constraints are satisfied by any orientation that can be given by algorithm. To do this, we will use a technique that extends the one used in \cite{Barth2022} to search for the existence of a feasible solution. 
	
	In this algorithm, for each possible orientation $(u, v)$ of an edge, the algorithm computes two values: the maximum flow that can be produced by the sources and that can be sent to $v$ through $u$, and the minimum flow that can be demanded by the sinks and that should come from $u$ through $v$. If the former is greater than the latter, there exists a feasible orientation. This is done with a dynamic programming algorithm using the tree structure of the graph. In the optimization variant of the problem, we should search for the same two values (computed with the real flow also with an extension of the dynamic programming algorithm of \cite{Barth2022}) with the additional constraint that the rounded loads of the sources lie in an interval $[\round{m}; \round{M}]$. If such values exist, we know there exists a feasible solution where the rounded load reserve is at most $\round{M} - \round{m}$. By enumerating all the possible couples of rounded loads, we get the desired feasible solution.
\end{itemize}

Note that this rounding method has a real physical meaning, since it is not possible to manipulate an electrical flow with an infinite precision. Thus a flow that is close to $2$ will be measured as 2 and thus rounded. Consequently, the approximate solution can be seen as the best solution that could be obtained with regard to the physical constraints of the measurement.

This section is divided into three parts. The first part defines the rounded flow and proves some useful properties about it. The second part proves the approximation ratio of the algorithms. The last part is dedicated to the description of a polynomial time algorithm that optimizes the rounded loads of the sources.

\begin{definition}
	Throughout this section, we denote with $n$ the size of $T$; with $\Pi$ the sum $\sum_{p \in P} \power{p}$; and with $\mathcal{N}$ an arbitrary numbering of the nodes of $T$. 
\end{definition}	
	
	\subsection{Rounding the flow and computing the rounding error}
	
	The goal of rounding the flow is to manipulate a polynomial number of possible values in each arc. 
	
	Without rounding, the flow is a rational between $1 / n^n$ (as the flow can be divided by at most $n$ at most $n$ times) and $\Pi$ (which roughly leads to an exponential number $O(n^n\Pi)$ of possible values for the flow). A naive rounding method would consist in fixing a value $\varepsilon > 0$ and rounding to the closest multiple of $\varepsilon / n$. We can show that, as the rounding error propagates, it never exceeds $\varepsilon$ which leads to an $\varepsilon$-approximation ratio of the loads. However, this decreases the number of values to $\varepsilon \Pi$ which is only pseudo-polynomial. 
	
	To overcome this obstacle, one should note that there is no need to round high flows to close values. Remember that we are interested in optimizing the loads of the sources, which means that the rounding error of a high rounded flow will be divided by the production of sources feeding that called power. And a high power should be necessarily powered by a source with a high production due to the capacity constraint. Consequently, the greater the flow value, the greater the rounding error can be. 
	
	We now fix two values $\varepsilon' \in ]0; 1/2[$ and $\varepsilon = \varepsilon' / (n^2+1)^2$ and define hereinafter the rounded flow, the rounded minimum load and the rounded load reserve. 
	
	\begin{definition}[Rounding a rational]
		\label{def:arrondi}
		Let $f > 0$ be a rational. If $f \in [2^i; 2^{i+1}[$ then we set the error $e(f) = 2^i \varepsilon$. We define $a(f)$ as the closest multiple of $e(f)$ lower than $f$, that is $a(f) = \left\lfloor \frac{f}{e(f)} \right\rfloor \cdot e(f)$.
	\end{definition}
	
	\begin{lemma}
		\label{lem:approx:1}
		$a(f) \geq f \cdot (1 - \varepsilon)$.
	\end{lemma}
	\begin{proof}
		Let $i$ such that $f \in [2^i; 2^{i+1}[$, then $a(f) = \left\lfloor \frac{f}{e(f)} \right\rfloor \cdot e(f) \geq f - e(f) \geq f - \varepsilon 2^i \geq f - \varepsilon \cdot f$.
	\end{proof}
	
	We now define the operator $\oplus$ that is the rounding version of $\sum_{i = 1}^p f_i / d$ and will be used in the formula of the rounded flow. 
	
	\begin{definition}
		\label{def:arrondi:2}
		Let $F = (f_1, f_2, \dots, f_p)$ be rationals and $d$ be an positive integer. Then 
		
		$$\oplus(F, d) =
		\begin{cases}
			0 &\text{if } p = 0\\
			a(\oplus((f_1, f_2, \dots, f_{p - 1}), d) +  \frac{f_p}{d}) &\text{otherwise}
		\end{cases}$$
	\end{definition}

	\begin{remark}
		\label{rem:arrondi}
		One may ask why we use this recursive formula leading to $a(a( \cdots a(a(\frac{f_1}{d}) + \frac{f_2}{d}) + \cdots ) + \frac{f_p}{d})$ that accumulates many rounding errors where $a(\sum_{i = 1}^p f_i / d)$ seems to do the same job with less errors. This has to do with the complexity of the approximation algorithm we will describe later. It is polynomial if we use the recursive formula and exponential if we use the latter formula. 
	\end{remark}

	\begin{lemma}
		\label{lem:approx:2}
		$\oplus(F, d) \geq \frac{1}{d} \sum\limits_{i = 1}^p f_i \cdot (1 - \varepsilon)^p$.
	\end{lemma}
	\begin{proof}
		The proof is done by induction on $p$. If $p = 0$, the two values equal 0.
		
		If we assume that $\oplus((f_1, f_2, \dots, f_{p - 1}), d) \geq \frac{1}{d}\sum\limits_{i = 1}^{p - 1} f_i \cdot (1 - \varepsilon)^{p - 1}$. 
		
		Then \begin{align*}
			\oplus(F, d) &= a(\oplus((f_1, f_2, \dots, f_{p - 1}), d) + \frac{f_p}{d})\\
			\intertext{By Lemma~\ref{lem:approx:1},}
			\oplus(F, d) &\geq \left(\oplus((f_1, f_2, \dots, f_{p - 1}), d) + \frac{f_p}{d}\right) \cdot (1 - \varepsilon)\\
			\oplus(F, d) &\geq (\frac{1}{d}\sum\limits_{i = 1}^{p - 1} f_i \cdot (1 - \varepsilon)^{p - 1} + \frac{f_p}{d}) \cdot (1 - \varepsilon)\\
			\intertext{As $(1 - \varepsilon)^p \leq 1 - \varepsilon$}
			\oplus(F, d) &\geq (\frac{1}{d}\sum\limits_{i = 1}^{p - 1} f_i + \frac{f_p}{d})  \cdot (1 - \varepsilon)^p
		\end{align*}
	\end{proof}
	
	We can now define the rounded flow. 
	
	\begin{definition}[Rounded flow]
		\label{def:roundedflow}
		Let $\mathcal{O}$ be a feasible orientation. 
		
		Let $d^-(v)$ be the in-degree of $v$ and $(v_1, v_2, \dots, v_p)$ be the set of successors of $v$, ordered using the numbering $\mathcal{N}$. The rounded flow $\roundedflow{\mathcal{O}}{u, v}$ going through each arc $(u, v)$
		 entering $v$ is  
		\begin{itemize}
			\item if $v \in W$, $\roundedflow{\mathcal{O}}{u, v} =  \oplus((\roundedflow{\mathcal{O}}{v, v_i})_{i \in \llbracket 1; p \rrbracket}, d^-(v))$ if $d^-(v) \neq 0$ and $+\infty$ otherwise
			\item if $v \in S$, $\roundedflow{\mathcal{O}}{u, v} =  \oplus((\roundedflow{\mathcal{O}}{v, v_i})_{i \in \llbracket 1; p \rrbracket}, d^-(v) + 1)$
			\item if $v \in P$, $\roundedflow{\mathcal{O}}{u, v} =  \oplus((\power{v}, \roundedflow{\mathcal{O}}{v, v_i})_{i \in \llbracket 1; p \rrbracket}, d^-(v))$ if $d^-(v) \neq 0$ and $+\infty$ otherwise 
		\end{itemize}
	
		We also denote by $\roundedflow{\mathcal{O}}{v}$ this value as $u$ does not intervene in the formula.
		
	\end{definition}
	
	\begin{example}
		We reuse the network given in Figure~\ref{fig:example}. We show in Figure~\ref{fig:example:roundedflow} the rounded flow of the orientations given in Figures~\ref{fig:example:b} and \ref{fig:example:d}. We use the parameter $\varepsilon' = 0.1$. As $n = 8$, we have $\varepsilon = \varepsilon' / (n^2 + 1)^2 \simeq 2.36 \cdot 10^{-5}$. 
		
		We assume the numbering $\mathcal{N}$ of the nodes of $T$ gives the following order: $(s_1, s_2, w_1, w_2, w_3, p_1, p_2, p_3)$.
		
		On the two figures, the rounded flow $\roundedflow{\mathcal{O}}{s_1, p_3}$ is $\oplus((\power{(p_3)}), 1) = a(\power{(p_3)}) = a(10)$. The value $10$ is rounded to the highest multiple of $2^3 \cdot \varepsilon \simeq 1.89 \cdot 10^{-4}$ lower than $10$, which is $2^3 \cdot \varepsilon \cdot 52812 \simeq 9.99991$.
		We can also compute $\roundedflow{\mathcal{O}}{w_3, p_2} = a(20)$, the highest multiple of $2^4 \cdot \varepsilon$ which is around $19.99981$ and $\roundedflow{\mathcal{O}}{w_1, p_1} = a(50)$, the highest multiple of $2^5 \cdot \varepsilon$ which is around $49.99952$.
	
		On Figure~\ref{fig:example:roundedflow:a}, $\roundedflow{\mathcal{O}}{s_1, w_1} = a(\roundedflow{\mathcal{O}}{w_1, p_1}) = a(a(50))$ is the highest multiple of $2^5 \cdot \varepsilon$ lower than $a(50)$ which is exactly $a(50) \simeq 49.99952$. The same occurs for $\roundedflow{\mathcal{O}}{s_2, p_3} = 19.99981$. The rounded flow in $s_1$ is  
		\begin{align*}
		\roundedflow{\mathcal{O}}{s_1} &= \oplus((\roundedflow{\mathcal{O}}{s_1, w_1}, \roundedflow{\mathcal{O}}{s_1, p_3}), 1)\\
		&= a(a(\roundedflow{\mathcal{O}}{s_1, w_1}) + \roundedflow{\mathcal{O}}{s_1, p_3})\\
		&= a(a(a(a(50))) + a(10))\\
		&\simeq 59.99943	
		\end{align*}
		Similarly, $\roundedflow{\mathcal{O}}{s_2} = 19.99981$.
		
		On Figure~\ref{fig:example:roundedflow:a}, the same reasoning leads to $\roundedflow{\mathcal{O}}{s_2} \simeq 9. 99991$ and 
		\begin{align*}
			\roundedflow{\mathcal{O}}{s_1} &= \oplus((\roundedflow{\mathcal{O}}{s_1, w_1}, \roundedflow{\mathcal{O}}{s_1, p_3}), 1)\\
			&= a(a(\roundedflow{\mathcal{O}}{s_1, w_1}) + \roundedflow{\mathcal{O}}{s_1, p_3})\\
			&\simeq a(69.99934)\\
			&\simeq 69.99858	
		\end{align*}

		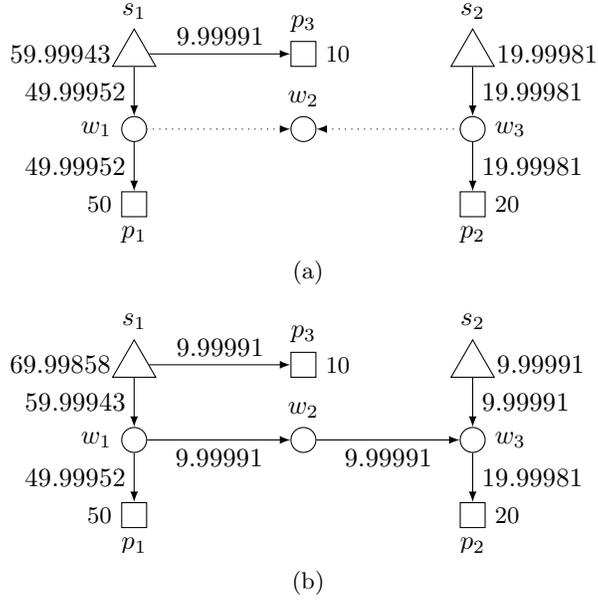
\begin{figure}[!ht]
			\captionsetup[subfigure]{justification=centering}
			\center
	
			\begin{subfigure}[t]{\textwidth}
				\begin{tikzpicture}[xscale=1.5]
					\clip (-5.5,-1.5) rectangle (3.75, 2);
					
					\node[source] [label=above:$s_1$, label=left:$59.99943$] (S1) at (-3,1) {};
					\node[source] [label=above:$s_2$, label=right:$19.99981$] (S2) at (0,1) {};
					
					\node[switch] [label=left:$w_1$] (W1) at (-3,0) {};
					\node[switch] [label=above:$w_2$] (W12) at (-1.5,0) {};
					\node[switch] [label=right:$w_3$] (W2) at (0,0) {};
					
					\node[puits] [label=below:$p_1$, label=left:{\small$50$}] (P1) at (-3,-1) {};
					\node[puits] [label=below:$p_{2}$, label=right:{\small$20$}] (P2) at (0,-1) {};
					\node[puits] [label=above:$p_{3}$, label=right:{\small$10$}] (P3) at (-1.5,1) {};
					
					\draw[->, >=latex] (S1) -- (W1) node[midway, left] {49.99952};
					\draw[->, >=latex] (S1) -- (P3) node[midway, above] {9.99991};
					\draw[->, >=latex] (S2) -- (W2) node[midway, right] {19.99981};
					
					\draw[->, >=latex] (W1) -- (P1)  node[midway, left] {49.99952};
					\draw[<-, >=latex, dotted] (W12) -- (W1);
					\draw[<-, >=latex, dotted] (W12) -- (W2);
					\draw[->, >=latex] (W2) -- (P2)  node[midway, right] {19.99981};
				\end{tikzpicture}
				\caption{}
				\label{fig:example:roundedflow:a}
			\end{subfigure}
			\begin{subfigure}[t]{\textwidth}
				\begin{tikzpicture}[xscale=1.5]
					\clip (-5.5,-1.5) rectangle (3.75, 2);
					
					\node[source] [label=above:$s_1$, label=left:$69.99858$] (S1) at (-3,1) {};
					\node[source] [label=above:$s_2$, label=right:$9.99991$] (S2) at (0,1) {};
					
					\node[switch] [label=left:$w_1$] (W1) at (-3,0) {};
					\node[switch] [label=above:$w_2$] (W12) at (-1.5,0) {};
					\node[switch] [label=right:$w_3$] (W2) at (0,0) {};
					
					\node[puits] [label=below:$p_1$, label=left:{\small$50$}] (P1) at (-3,-1) {};
					\node[puits] [label=below:$p_{2}$, label=right:{\small$20$}] (P2) at (0,-1) {};
					\node[puits] [label=above:$p_{3}$, label=right:{\small$10$}] (P3) at (-1.5,1) {};
					
					\draw[->, >=latex] (S1) -- (W1) node[midway, left] {59.99943};
					\draw[->, >=latex] (S1) -- (P3) node[midway, above] {9.99991};
					\draw[->, >=latex] (S2) -- (W2) node[midway, right] {9.99991};
					
					\draw[->, >=latex] (W1) -- (P1)  node[midway, left] {49.99952};
					\draw[<-, >=latex] (W12) -- (W1) node[midway, below] {9.99991};
					\draw[->, >=latex] (W12) -- (W2) node[midway, below] {9.99991};
					\draw[->, >=latex] (W2) -- (P2)  node[midway, right] {19.99981};
				\end{tikzpicture}
				\caption{}
				\label{fig:example:roundedflow:b}
			\end{subfigure}
			\caption{Example of rounded flow. Next to each arc and source is written the rounded flow going through it.}
			\label{fig:example:roundedflow}
			\end{figure}
	\end{example}

	Recall that we set a constant $\varepsilon'$ and fixed $\varepsilon = \varepsilon' / (n^2+1)^2$. We now prove that the rounding error that is propagated through the arcs of the tree does not lead to a mistake greater than $(1 - \varepsilon')$.

	\begin{lemma}
		\label{lem:roundingerror}
		Let $\mathcal{O}$ be a feasible orientation and $(u, v)$ be an arc of the directed tree, then $\roundedflow{\mathcal{O}}{u, v} \in [\flow{\mathcal{O}}{u, v} \cdot (1 - \varepsilon'); \flow{\mathcal{O}}{u, v}]$.
	\end{lemma}
	\begin{proof}
		The upper bound is trivial as we round down the flow. We prove the lower bound by induction on $p$, where $p$ is the length of a longest path from $v$ to a sink, that $\roundedflow{\mathcal{O}}{u, v} \geq \flow{\mathcal{O}}{u, v} \cdot (1 - \varepsilon)^{np + 1}$. Note that, if such a sink does not exist, then the flow and the rounded flow going out of $v$ is 0.
		
		If $p = 0$ then $v$ is a sink and there is no path from $v$ to another sink. In that case, $\roundedflow{\mathcal{O}}{u, v} = \oplus((\power{v}), d^-(v)) = a(\power{v} / d^-(v)) = a(\flow{\mathcal{O}}{u, v}) \geq \flow{\mathcal{O}}{u, v} \cdot (1 - \varepsilon)$ by Lemma~\ref{lem:approx:1}. 
		
		We now assume that $p > 0$ and that the property is true for every node at distance at most $p - 1$ from a sink. Let $v$ be a node for which the farthest sink is at distance $p$ and let $(v_1, v_2, \dots, v_q)$ be the successors of $v$, ordered using the numbering $\mathcal{N}$. The longest path from $v_i$ to a sink is $p_i \leq p - 1$. By induction, $\roundedflow{\mathcal{O}}{v, v_i} \geq \flow{\mathcal{O}}{v, v_i} \cdot (1 - \varepsilon)^{np_i + 1} \geq \flow{\mathcal{O}}{v, v_i} \cdot (1 - \varepsilon)^{n \cdot(p - 1) + 1}$. 
		
		We also assume without loss of generality that $v \in W$. The calculations can be easily adapted to the cases where $v \in P$ or $v \in S$. Let 
		\begin{align*}
			\roundedflow{\mathcal{O}}{u, v} &= \oplus((\roundedflow{\mathcal{O}}{v, v_i})_{i \in \llbracket 1; p \rrbracket}, d^-(v))
			\intertext{By Lemma~\ref{lem:approx:2}}
			&\geq \frac{1}{d^-(v)} \sum\limits_{i = 1}^q \roundedflow{\mathcal{O}}{v, v_i} \cdot (1 - \varepsilon)^{q}\\
			\intertext{As $q \leq n$}
			&\geq \frac{1}{d^-(v)} \sum\limits_{i = 1}^q \roundedflow{\mathcal{O}}{v, v_i} \cdot (1 - \varepsilon)^{n}\\
			&\geq \frac{1}{d^-(v)} \sum\limits_{i = 1}^q \flow{\mathcal{O}}{v, v_i} \cdot (1 - \varepsilon)^{n \cdot(p - 1) + 1} \cdot (1 - \varepsilon)^{n}\\
			&\geq \frac{1}{d^-(v)} \sum\limits_{i = 1}^q \flow{\mathcal{O}}{v, v_i} \cdot (1 - \varepsilon)^{np + 1}\\
			&\geq \flow{\mathcal{O}}{u, v} \cdot (1 - \varepsilon)^{np + 1}
		\end{align*}
		
		In addition, as $p \leq n$
		
		\begin{align*}
			(1 - \varepsilon)^{pn + 1} &\geq (1 - \varepsilon)^{n^2 + 1}\\ 
			\intertext{As $(1 - \varepsilon)^{n^2 + 1} = \sum\limits_{i = 0}^{n^2 + 1}\binom{n^2+1}{i}(-\varepsilon)^i$}
			(1 - \varepsilon)^{pn + 1} &\geq \sum\limits_{i = 0}^{n^2 + 1}\binom{n^2+1}{i}(-\varepsilon)^i\\ 
			(1 - \varepsilon)^{pn + 1} &\geq 1 - \sum\limits_{i = 1}^{n^2 + 1}\binom{n^2+1}{i}\varepsilon^i\\ 
			\intertext{As $\binom{a}{b} \leq a^b$}
			(1 - \varepsilon)^{pn + 1} &\geq 1 - \sum\limits_{i = 1}^{n^2 + 1}(\varepsilon \cdot (n^2 + 1))^i\\
			\intertext{As $\varepsilon \cdot (n^2 + 1) = \frac{\varepsilon'}{n^2 + 1}$}
			(1 - \varepsilon)^{pn + 1} &\geq 1 - \sum\limits_{i = 1}^{n^2 + 1}\left(\frac{\varepsilon'}{n^2 + 1}\right)^i\\
			(1 - \varepsilon)^{pn + 1} & \geq 1 - (n^2 + 1) \frac{\varepsilon'}{n^2 + 1} = 1 - \varepsilon'
		\end{align*}
	\end{proof}
	
%	\begin{example}
%		The highest error on all the rounded flows in Figure~\ref{fig:example:roundedflow} is made by $\roundedflow{\mathcal{O}}{s_1}$ in Figure~\ref{fig:example:roundedflow:b}. The flow is $70$ and the rounded flow is around $69.99858$. We have $\roundedflow{\mathcal{O}}{s_1} / \flow{\mathcal{O}}{s_1} \geq 0.99997$. The error rate is then at most $3 \cdot 10^-5 < \varepsilon' = 0.1$.
%	\end{example}

	\subsection{Approximation factor}
	
	With the rounded flow, we can now define the rounded loads, and thus a rounded version of $\minload{\mathcal{O}}$, $\maxload{\mathcal{O}}$ and $\loadreserve{\mathcal{O}}$.
	
	\begin{definition}[Rounded objectives]
		\label{def:roundedobj}
		Given a feasible orientation $\mathcal{O}$, the rounded load of a source $s \in S$ is defined as the ratio $\roundedload{\mathcal{O}}{s} = \frac{\roundedflow{\mathcal{O}}{s}}{\production{s}}$. 
		
		\begin{itemize}
			\item The \emph{minimum rounded load} $\roundedminload{\mathcal{O}}$ is $\min\limits_{s \in S} \roundedload{\mathcal{O}}{s}$.
			\item The \emph{maximum rounded load} $\roundedmaxload{\mathcal{O}}$ is $\max\limits_{s \in S} \roundedload{\mathcal{O}}{s}$.
			\item The \emph{rounded load reserve} $\roundedloadreserve{\mathcal{O}}$ is $\roundedmaxload{\mathcal{O}} - \roundedminload{\mathcal{O}}$
		\end{itemize}
		
	\end{definition}

	In this part we prove that a feasible solution minimizing the rounded load reserve is a close approximation of the solutions minimizing the load reserve. Similarly, a solution that maximizes the minimum rounded load approximates the solutions that maximize the minimum load. 
	
	Note that rounding the flow may change the set of feasible orientations as such a flow is lower than the regular flow $\flowst{\mathcal{O}}$ and thus may be lower than some capacity exceeded by the flow. However, hereinafter, by \emph{feasible solution}, we mean a solution satisfying the capacity constraints and the demand constraints with the regular flow as explained in Definition~\ref{def:valid:3}, not with the rounded flow $\roundedflowst{\mathcal{O}}$, which only changes the metric used to measure the quality of the feasible solutions. We show in the next subsection how to build solutions with good rounded metrics in polynomial time even though we use the regular flow to check the feasibility of the returned orientation. 
	
	\begin{lemma}
		\label{lem:approx:maxminm}
		Let $\mathcal{O}^*$ be an orientation maximizing $\minload{\mathcal{O}}$ and $\mathcal{\round{O}}$ a feasible solution maximizing $\roundedminload{\mathcal{O}}$ then $\minload{\mathcal{\round{O}}} \geq \minload{\mathcal{O}^*} \cdot (1 - \varepsilon')$ and $\minload{\mathcal{\round{O}}} \geq \minload{\mathcal{O}^*} - \varepsilon'$.
	\end{lemma}
	\begin{proof}
	Let $s_m$ be the source with the lower (not rounded) load in $\mathcal{\round{O}}$ and $s_m^*$ be the source with the lower rounded load in $\mathcal{O}^*$, in other words: $\minload{\mathcal{\round{O}}} = \frac{\flow{\mathcal{\round{O}}}{s_m}}{\production{s_m}}$ and $\roundedminload{\mathcal{O}^*} = \frac{\roundedflow{\mathcal{O^*}}{s_m^*}}{\production{s_m^*}}$.
	
	By Lemma~\ref{lem:roundingerror} and definition of $\roundedminload{\mathcal{\round{O}}}$,

	\begin{align*}
		\minload{\mathcal{\round{O}}} &= \frac{\flow{\mathcal{\round{O}}}{s_m}}{\production{s_m}} \geq \frac{\roundedflow{\mathcal{\round{O}}}{s_m}}{\production{s_m}} \geq \roundedminload{\mathcal{\round{O}}}\\
		\intertext{As, by definition of $\mathcal{\round{O}}$, it maximizes $\roundedminload{\mathcal{\round{O}}}$, then}
	\roundedminload{\mathcal{\round{O}}} &\geq  \roundedminload{\mathcal{O}^*} = \frac{\roundedflow{\mathcal{O^*}}{s_m^*}}{\production{s_m^*}}\\
		\intertext{By Lemma~\ref{lem:roundingerror} and definition of $\minload{\mathcal{O}^*}$, }
	\roundedminload{\mathcal{O}^*} &\geq  \frac{\flow{\mathcal{O^*}}{s_m^*}}{\production{s_m^*}} \cdot (1 - \varepsilon') \geq \minload{\mathcal{O}^*} \cdot (1 - \varepsilon')
		\intertext{Finally, as $\minload{\mathcal{O}^*} \leq 1$, we also have}
		\minload{\mathcal{O}^*} \cdot (1 - \varepsilon') &\geq \minload{\mathcal{O}^*}  - \varepsilon'
	\end{align*}
	\end{proof}
	
	\begin{lemma}
		\label{lem:approx:minr}
		Let $\mathcal{O}^*$ be an orientation minimizing $R(\mathcal{O})$ and $\mathcal{\round{O}}$ a feasible solution minimizing $\roundedloadreserve{\mathcal{O}}$ then $\loadreserve{\mathcal{\round{O}}} \leq \loadreserve{\mathcal{O}^*} + 3 \varepsilon'$.
	\end{lemma}
	\begin{proof}
		Let $s_m$ and $s_M$ be the sources with respectively the lower and higher (not rounded) loads in $\mathcal{\round{O}}$ and let $s_m^*$ and $s_M^*$ be the sources with respectively the lower and higher rounded loads in $\mathcal{O}^*$, in other words:
		
		\begin{itemize}
			\item $\maxload{\mathcal{\round{O}}} = \frac{\flow{\mathcal{\round{O}}}{s_M}}{\production{s_M}}$ and $\minload{\mathcal{\round{O}}} = \frac{\flow{\mathcal{\round{O}}}{s_m}}{\production{s_m}}$.
			\item $\roundedmaxload{\mathcal{O}^*} = \frac{\roundedflow{\mathcal{O^*}}{s_M^*}}{\production{s_M^*}}$ and $\roundedminload{\mathcal{O}^*} = \frac{\roundedflow{\mathcal{O^*}}{s_m^*}}{\production{s_m^*}}$
			
		\end{itemize}

		\begin{align*}
			\loadreserve{\mathcal{\round{O}}} &= \maxload{\mathcal{\round{O}}} - \minload{\mathcal{\round{O}}} =  \frac{\flow{\mathcal{\round{O}}}{s_M}}{\production{s_M}} - \frac{\flow{\mathcal{\round{O}}}{s_m}}{\production{s_m}}\\
			\intertext{By Lemma~\ref{lem:roundingerror}, }
			\loadreserve{\mathcal{\round{O}}} &\leq \frac{\roundedflow{\mathcal{\round{O}}}{s_M}}{\production{s_M}} \cdot \frac{1}{1 - \varepsilon'} - \frac{\roundedflow{\mathcal{\round{O}}}{s_m}}{\production{s_m}}\\
			\intertext{By definition of $\roundedmaxload{\mathcal{\round{O}}}$ and $\roundedminload{\mathcal{\round{O}}}$}
				\loadreserve{\mathcal{\round{O}}} &\leq \roundedmaxload{\mathcal{\round{O}}} \cdot \frac{1}{1 - \varepsilon'} - \roundedminload{\mathcal{\round{O}}}\\
			\intertext{As $\frac{1}{1 - \varepsilon'} = 1 + \frac{\varepsilon'}{1 - \varepsilon'} $ and $\roundedmaxload{\mathcal{\round{O}}} \leq 1$}
			\roundedmaxload{\mathcal{\round{O}}} \cdot \frac{1}{1 - \varepsilon'} - \roundedminload{\mathcal{\round{O}}} &\leq \roundedmaxload{\mathcal{\round{O}}} - \roundedminload{\mathcal{\round{O}}} + \frac{\varepsilon'}{1 - \varepsilon'} = \roundedloadreserve{\mathcal{\round{O}}} + \frac{\varepsilon'}{1 - \varepsilon'}
			\intertext{As, by definition of $\mathcal{\round{O}}$, it minimizes $\roundedloadreserve{\mathcal{\round{O}}}$, then}
			\roundedloadreserve{\mathcal{\round{O}}} + \frac{\varepsilon'}{1 - \varepsilon'} \leq \roundedloadreserve{\mathcal{O}^*} + \frac{\varepsilon'}{1 - \varepsilon'} &= \frac{\roundedflow{\mathcal{O^*}}{s_M^*}}{\production{s_M^*}} - \frac{\roundedflow{\mathcal{O^*}}{s_m^*}}{\production{s_m^*}} + \frac{\varepsilon'}{1 - \varepsilon'}
			\intertext{By Lemma~\ref{lem:roundingerror}}
			\roundedloadreserve{\mathcal{O}^*} + \frac{\varepsilon'}{1 - \varepsilon'} & \leq \frac{\flow{\mathcal{O^*}}{s_M^*}}{\production{s_M^*}} - \frac{\flow{\mathcal{O^*}}{s_m^*}}{\production{s_m^*}} \cdot (1 - \varepsilon') + \frac{\varepsilon'}{1 - \varepsilon'}
			\intertext{By definition of $\maxload{\mathcal{O}^*}$ and $\minload{\mathcal{O}^*}$}
			\roundedloadreserve{\mathcal{O}^*} + \frac{\varepsilon'}{1 - \varepsilon'} & \leq \maxload{\mathcal{O}^*} - \minload{\mathcal{O}^*} \cdot (1 - \varepsilon') + \frac{\varepsilon'}{1 - \varepsilon'}
			\intertext{Finally, as $\varepsilon' \in ]0; \frac{1}{2}[$, then $\frac{\varepsilon'}{1 - \varepsilon'} \leq 2\varepsilon'$, and as $\minload{\mathcal{O}^*} \leq 1$}
			\maxload{\mathcal{O}^*} - \minload{\mathcal{O}^*} \cdot (1 - \varepsilon') + \frac{\varepsilon'}{1 - \varepsilon'} &\leq \maxload{\mathcal{O}^*} - \minload{\mathcal{O}^*} + 3\varepsilon'
		\end{align*}
	\end{proof}

	\subsection{Compute a feasible orientation minimizing $\roundedloadreserve{\mathcal{O}}$ or maximizing $\roundedminload{\mathcal{O}}$ in polynomial time}
	
	Recall that we denote by $n$ the size of the tree $T$ and by $\Pi = \sum_{p \in P} \power{p}$.  We prove it is possible to return in polynomial time the approximate solutions mentioned in Lemmas~\ref{lem:approx:maxminm} and \ref{lem:approx:minr}.
	
	To do so, we adapt the algorithm of \cite{Barth2022}, which solved \valid in polynomial time, to those problems.

	We first prove the following lemma that states the number of distinct values for a rounded flow is bounded by a polynomial.
	\begin{lemma}
		\label{lem:nbflots}
		Let $$\round{F} = \{\roundedflow{\mathcal{O}}{\mathcal{O}([u, v])}, (u, v) \in T, \mathcal{O} \text{ is a feasible orientation}\}$$
		Then, $\lvert \round{F}\rvert  \leq (n\log(n) + \log(\Pi) + 1) (1 + \frac{1}{\varepsilon}) + 1 = (n\log(n) + \log(\Pi) + 1) (1 + \frac{(n+1)^2}{\varepsilon'}) + 1 = O(\frac{1}{\varepsilon'} \cdot (n^3\log(n) + n^2\log(\Pi)))$.
	\end{lemma}
	\begin{proof}
		
		Let $f \in [2^i; 2^{i + 1}[$ then $a(f)$ is the closest multiple of $2^i \varepsilon$ lower than $f$. Those multiples belong to the interval $[\left\lfloor \frac{2^i}{2^i \varepsilon} \right\rfloor \cdot 2^i \varepsilon; 2^{i + 1}] \subset [2^i (1 - \varepsilon); 2^{i + 1}]$. Then there  are at most  $(2^{i + 1} - 2^i \cdot (1 - \varepsilon)) / 2^i \varepsilon = 1 + 1/\varepsilon$ such multiples. As a consequence, if the value $i$ is between $a$ and $b$ then, there are at most $(1 + 1/\varepsilon) \cdot (b - a)$ possible values for the rounded flow.
		
		Let $\mathcal{O}$ be a feasible orientation, and $(u, v)$ be an arc in that orientation, then $\roundedflow{\mathcal{O}}{u, v} \leq \flow{\mathcal{O}}{u, v} \leq \Pi$. Thus $i \leq \log(\Pi)$. 
		
		The flow can also be split. This happens at most once per node in the tree and each node has at most $n$ entering arcs, thus the flow $\flow{\mathcal{O}}{u, v}$ cannot be less than $1/n^n$ unless if the flow is zero. By Lemma~\ref{lem:roundingerror}, $\roundedflow{\mathcal{O}}{u, v}  \geq (1 - \varepsilon') / n^n$. We recall that $\varepsilon'$ was chosen between 0 and $1/2$, thus $1 - \varepsilon' \geq 1/2$ and, then, $(1 - \varepsilon')/n^n \geq  1 / (2 \cdot n^n) = 2^{-n\log(n) - 1}$. Consequently, $i \geq - n\log(n) - 1$.
		
		Thus, the number of non null rounded flows is $(1 + 1/\varepsilon) \cdot (\log(\Pi) - (- n\log(n) - 1))$. There is one additional possible rounded flow: the null flow. 
	\end{proof}

\subsubsection{Feasible semi-orientations}

Given an edge $[u, v]$, we consider the two trees $T_u$ and $T_v$ of the forest $T \backslash [u, v]$, respectively containing $u$ and $v$. 

We consider a feasible orientation $\mathcal{O}$ of $T$, where $\mathcal{O}([u, v]) = (u, v)$. The flow $\flowst{\mathcal{O}}$ and the rounded flow $\roundedflowst{\mathcal{O}}$ in $T_v$ and $(u, v)$ do not depend on the orientation of $T_u$. Similarly, the flow and the rounded flow in $T_u$ depend only on the flow and the rounded flow in $(u, v)$ and on the orientation of $T_u$, but not on the orientations of $T_v$. We express this idea with the notion of semi-orientations. 

\begin{definition}
	\label{def:halforientation:outgoing}
	An \emph{semi-orientation $\mathcal{HO}$ outgoing from $(u, v)$ }is an orientation of the edges of $T_v$ and of $[u, v]$ such that $\mathcal{HO}([u, v])= (u, v)$. The calculations of the flow $\flowst{\mathcal{HO}}$ and rounded flow $\roundedflowst{\mathcal{HO}}$ in $T_v \cup \{[u, v]\}$ follow the Definitions~\ref{def:valid:3} and \ref{def:roundedflow}. We say $\mathcal{HO}$ is feasible if it satisfies the demand and the capacity constraints on the nodes of $T_v$.
\end{definition}

\begin{definition}
	\label{def:halforientation:entering}
	An \emph{semi-orientation $\mathcal{HO}$ entering $(u, v)$} is an orientation of the edges of $T_u$ and of $[u, v]$ such that $\mathcal{HO}([u, v])= (u, v)$. Given $f \geq 0$ and $\round{f} \in \round{F}$, if we set $\flow{\mathcal{HO}}{u, v} = f$ and $\roundedflow{\mathcal{HO}}{u, v} = \round{f}$ then we can compute the flow $\flowst{\mathcal{HO}}$ and rounded flow $\roundedflowst{\mathcal{HO}}$ in $T_u \cup \{[u, v]\}$ with Definitions~\ref{def:valid:3} and \ref{def:roundedflow}. We denote by \emph{semi-orientation $\mathcal{HO}$ entering $(u, v)$ with $f$ and $\round{f}$} the triplet $(\mathcal{HO}, f, \round{f})$ and we write $\flowst{(\mathcal{HO}, f, \round{f})}$ and $\roundedflowst{(\mathcal{HO}, f, \round{f})}$ the associated flows. We say $(\mathcal{HO}, f, \round{f})$ is feasible if and only if $\flowst{(\mathcal{HO}, f, \round{f})}$ satisfies the demand and the capacity constraints on the nodes of $T_u$. 
\end{definition}

\begin{remark}
	Considering a feasible semi-orientation $\mathcal{HO}$ outgoing from $(u, v)$, there is no constraint on the flow $\flow{\mathcal{HO}}{u, v}$: the demand and the capacity constraints should only be satisfied for the nodes of $T_v$.
\end{remark}

\begin{remark}
	Considering an semi-orientation entering $(u, v)$ with $f$ and $\tilde{f}$, then, even if $v$ is a sink, there is no relation between $f$, $\round{f}$ and $\power{v}$. 
\end{remark}

\begin{lemma}
	\label{lem:approx:3}
	Given $0 \leq f' < f$ and $\round{f} \in \round{F}$. If $(\mathcal{HO}, f, \round{f})$ is feasible, then $(\mathcal{HO}, f', \round{f})$ is also feasible.
\end{lemma}
\begin{proof}
	$\flowst{(\mathcal{HO}, f, \round{f})}$ satisfies the demand and the capacity constraint on the nodes of $T_u$. As $f' \leq f$, for every arc in $T_u$, $\flowst{(\mathcal{HO}, f', \round{f})}(u, v) \leq \flowst{(\mathcal{HO}, f, \round{f})}(u, v)$, thus $\flowst{(\mathcal{HO}, f', \round{f})}$ also satisfies the demand and capacity constraints.
\end{proof}

\begin{example}
	In Figure~\ref{fig:halfexample}, we reuse the example of Figure~\ref{fig:example} to show examples of feasible and not feasible semi-orientations.
	
	\begin{figure}[!ht]
		\captionsetup[subfigure]{justification=centering}
		\center
		\begin{subfigure}[t]{0.4\textwidth}
			\begin{tikzpicture}
				\clip (-3.75,-1.5) rectangle (0.75, 2);
				
				\node[source] [label=right:$$]  [label=above:$$]  (S1) at (-3,1) {};
				\node[source] [label=left:$s_2$] [label=above:$20/20$] (S2) at (0,1) {};
				
				\node[switch] [label=left:$v$] (W1) at (-3,0) {};
				\node[switch] [label=above:$u$] (W12) at (-1.5,0) {};
				\node[switch]  (W2) at (0,0) {};
				
				\node[puits]  (P1) at (-3,-1) {};
				\node[puits] [label=left:$20$] (P2) at (0,-1) {};
				\node[puits]  (P3) at (-1.5,1) {};
				
				\draw[-, >=latex, dotted] (S1) -- (W1);
				\draw[-, >=latex, dotted] (S1) -- (P3);
				\draw[->, >=latex] (S2) -- (W2);
				
				\draw[-, >=latex, dotted] (W1) -- (P1);
				\draw[->, >=latex] (W12) -- (W1);
				\draw[<-, >=latex] (W12) -- (W2);
				\draw[->, >=latex] (W2) -- (P2);
			\end{tikzpicture}
			\caption{Example of semi-orientation entering $(u, v)$. With $f = 0$ and any value $\round{f}$, this semi-orientation is feasible. However, if $f > 0$, it is not feasible anymore as the capacity constraint is not satisfied for $s_2$.}
			\label{fig:halfexample:b}
		\end{subfigure}
		\begin{subfigure}[t]{0.4\textwidth}
			\begin{tikzpicture}
				\clip (-4.25,-1.5) rectangle (0.75, 2);
				
				\node[source] [label=right:$$]  [label=above:$5/100$]  (S1) at (-3,1) {};
				\node[source] [label=left:$$] [label=above:$ $] (S2) at (0,1) {};
				
				\node[switch] [label=above left:$v$]  [label=left:$55/60$] (W1) at (-3,0) {};
				\node[switch] [label=above:$u$] [label=below:$55/20$] (W12) at (-1.5,0) {};
				\node[switch]  (W2) at (0,0) {};
				
				\node[puits] [label=left:$50$] (P1) at (-3,-1) {};
				\node[puits] (P2) at (0,-1) {};
				\node[puits] [label=above:$10$] (P3) at (-1.5,1) {};
				
				\draw[<-, >=latex] (S1) -- (W1);
				\draw[->, >=latex] (S1) -- (P3);
				\draw[-, >=latex, dotted] (S2) -- (W2);
				
				\draw[->, >=latex] (W1) -- (P1);
				\draw[->, >=latex] (W12) -- (W1);
				\draw[-, >=latex, dotted] (W12) -- (W2);
				\draw[-, >=latex, dotted] (W2) -- (P2);
			\end{tikzpicture}
			\caption{Example of semi-orientation outgoing from $(u, v)$. This semi-orientation is feasible as all the nodes of $T_v$ satisfy the capacity and the demand constraints (even if $u$ does not).}
			\label{fig:halfexample:c}
		\end{subfigure}\\
		\begin{subfigure}[t]{0.4\textwidth}
			\begin{tikzpicture}
				\clip (-3.75,-1.5) rectangle (0.75, 2);
				
				\node[source] [label=right:$$]  [label=above:$$]  (S1) at (-3,1) {};
				\node[source] [label=left:$$] [label=above:$$] (S2) at (0,1) {};
				
				\node[switch] [label=left:$v$] (W1) at (-3,0) {};
				\node[switch] [label=above:$u$] (W12) at (-1.5,0) {};
				\node[switch]  (W2) at (0,0) {};
				
				\node[puits] [label=left:$p_1$] (P1) at (-3,-1) {};
				\node[puits]  (P2) at (0,-1) {};
				\node[puits]  (P3) at (-1.5,1) {};
				
				\draw[<-, >=latex] (S1) -- (W1);
				\draw[->, >=latex] (S1) -- (P3);
				\draw[-, >=latex, dotted] (S2) -- (W2);
				
				\draw[<-, >=latex] (W1) -- (P1);
				\draw[->, >=latex] (W12) -- (W1);
				\draw[-, >=latex, dotted] (W12) -- (W2);
				\draw[-, >=latex, dotted] (W2) -- (P2);
			\end{tikzpicture}
			\caption{This semi-orientation is not feasible as the demand constraint\\ is not satisfied for $p_1$}
			\label{fig:halfexample:e}
		\end{subfigure}%
		\begin{subfigure}[t]{0.4\textwidth}
			\begin{tikzpicture}
				\clip (-3.75,-1.5) rectangle (0.75, 2);
				
				\node[source] (S1) at (-3,1) {};
				\node[source] (S2) at (0,1) {};
				
				\node[switch] [label=left:$v$] (W1) at (-3,0) {};
				\node[switch] [label=above:$u$] (W12) at (-1.5,0) {};
				\node[switch]  (W2) at (0,0) {};
				
				\node[puits]  (P1) at (-3,-1) {};
				\node[puits]  (P2) at (0,-1) {};
				\node[puits]  (P3) at (-1.5,1) {};
				
				\draw[-, >=latex, dotted] (S1) -- (W1);
				\draw[-, >=latex, dotted] (S1) -- (P3);
				\draw[->, >=latex] (S2) -- (W2);
				
				\draw[-, >=latex, dotted] (W1) -- (P1);
				\draw[->, >=latex] (W12) -- (W1);
				\draw[->, >=latex] (W12) -- (W2);
				\draw[->, >=latex] (W2) -- (P2);
			\end{tikzpicture}
			\caption{This semi-orientation is not feasible as the demand constraint\\ is not satisfied for $u$.}
			\label{fig:halfexample:f}
		\end{subfigure}\\
		\caption{Example of instance and orientations for that instance. For each orientation, it is explained if it is feasible according to Definitions~\ref{def:halforientation:outgoing} and \ref{def:halforientation:entering}.}
		\label{fig:halfexample}
	\end{figure}
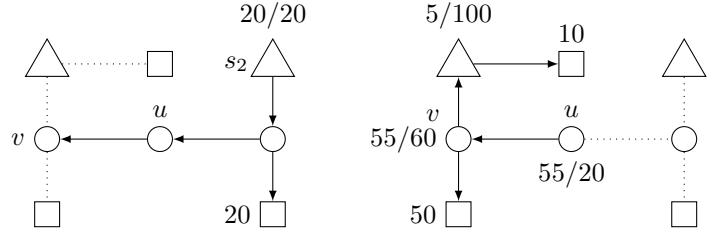
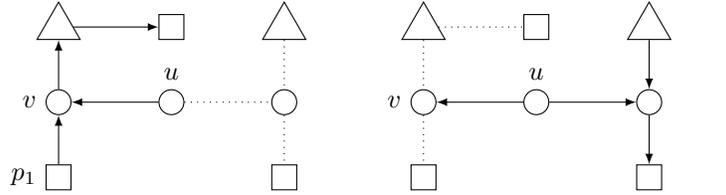
\end{example}

\subsubsection{Functions $\mininputflowst$ and $\maxoutputflowst$}

We adapt the notations $\roundedmaxload{\mathcal{O}}$ and $\roundedminload{\mathcal{O}}$ to the semi-orientations. 

Given an semi-orientation $\mathcal{HO}$ outgoing from $(u, v)$, $\roundedmaxload{\mathcal{HO}}$ and $\roundedminload{\mathcal{HO}}$ are the maximum and minimum rounded loads of the sources in $T_v$ (not including $u$).

Given an semi-orientation $\mathcal{HO}$ entering $(u, v)$ with $f$ and $\round{f}$, $\roundedmaxload{\mathcal{HO}, f, \round{f}}$ and $\roundedminload{\mathcal{HO}, f, \round{f}}$ are the maximum and minimum rounded loads of the sources in $T_u$ (not including $v$).

\begin{remark}
	In case of an semi-orientation $\mathcal{HO}$ outgoing from $(u, v)$ where $u$ is a source, as $u$ is not part of $T_v$, the load of $u$ does not intervene in the calculations of  $\roundedmaxload{\mathcal{HO}}$ and $\roundedminload{\mathcal{HO}}$. This is necessary as, otherwise, only the flow going from $u$ to $v$ would count in the load of $u$ and this would mistakenly reduce the value of $\roundedmaxload{\mathcal{HO}}$ and $\roundedminload{\mathcal{HO}}$.
	
	Similarly, if $v$ is a source and is considered in the definition of $\roundedmaxload{\mathcal{HO}, f, \round{f}}$ and $\roundedminload{\mathcal{HO}, f, \round{f}}$, where $\mathcal{HO}$ is an semi-orientation entering $(u, v)$, we would have $\roundedminload{\mathcal{HO}} = 0$ whatever the orientation of $T_u$ is.
\end{remark}

We write $\round{L}= \{\frac{\round{f}}{\production{s}} \vert  \round{f} \in \round{F}, s \in S\}$. Note that $\lvert \round{L}\rvert  = \lvert \round{F}\rvert  \cdot \lvert S\rvert $ is polynomial and that $\roundedmaxload{\mathcal{HO}}$, $\roundedminload{\mathcal{HO}}$, $\roundedmaxload{\mathcal{HO}, f, \round{f}}$ and $\roundedminload{\mathcal{HO}, f, \round{f}}$ are in $\round{L}$.

Building a feasible orientation minimizing the rounded load reserve can be done using the following algorithm: given any two neighbors $u$ and $v$ and three values $\round{f} \in \round{F}, \round{M} \in \round{L}, \round{m} \in \round{L}$, we search for a rational $f$ such that there exists a feasible semi-orientation $\mathcal{HO}$ outgoing from $(u, v)$ satisfying $\flow{\mathcal{HO}}{u, v} = f$ and $\roundedflow{\mathcal{HO}}{u, v} = \round{f}$ and a feasible semi-orientation $\mathcal{HO}'$ entering $(u, v)$ with $f$ and $\round{f}$. In addition, the loads should be such that $\roundedminload{\mathcal{HO}} \geq \round{m}$, $\roundedminload{\mathcal{HO}, f, \round{f}} \geq \round{m}$, $\roundedmaxload{\mathcal{HO}} \leq \round{M}$ and $\roundedmaxload{\mathcal{HO}, f, \round{f}} \leq \round{M}$.

If such semi-orientations exist, the union is a feasible orientation $\mathcal{O}$ such that $\mathcal{O}([u, v]) = (u, v)$, $\roundedminload{\mathcal{O}} \geq \round{m}$ and $\roundedmaxload{\mathcal{O}} \leq \round{M}$. We start again for every possible values of $\round{f}$, $\round{M}$, $\round{m}$. We then renew this operation by reversing the orientation of $[u, v]$. We get a feasible orientation minimizing the rounded load reserve by returning the one where $\round{M} - \round{m}$ is minimum, and a solution maximizing the minimum rounded load by returning the one where $\round{m}$ is maximum (in that case, enumerating all the values of $\round{M}$ is not necessary, we can set $\round{M}$ to 1).  

It is possible to enumerate the triplets $\round{f}, \round{m}$ and $\round{M}$ in polynomial time, but not all the possible values of $f$. This is why we use auxiliary functions $\mininputflowst$ and $\maxoutputflowst$ that will respectively give the lowest possible value $f$ for $\mathcal{HO}$ and the highest possible value $f$ for $\mathcal{HO}'$.

\begin{definition}
	\label{def:approx:2}
	
	Let $\round{f} \in \round{F}, \round{M} \in \round{L}, \round{m} \in \round{L}$ and $[u, v]$ be an edge of $T$. Let $\mathcal{P}$ be the set of all values $f$ such that there exists a feasible semi-orientation $\mathcal{HO}$ outgoing from $(u, v)$ such that $\flow{\mathcal{HO}}{u, v} = f$, $\roundedflow{\mathcal{HO}}{u, v} = \round{f}$, $\roundedmaxload{\mathcal{HO}} \leq \round{M}$ and $\roundedminload{\mathcal{HO}} \geq \round{m}$. We then write $\mininputflow{u}{v}{\round{f}}{\round{M}}{\round{m}} = \min(\mathcal{P})$. If $\mathcal{P} = \emptyset$, then $\mininputflow{u}{v}{\round{f}}{\round{M}}{\round{m}} = +\infty$.
\end{definition}

\begin{definition}
	\label{def:approx:1}
	Let $\round{f} \in \round{F}, \round{M} \in \round{L}, \round{m} \in \round{L}$ and $[u, v]$ be an edge of $T$. Let $\mathcal{P}$ be the set of all values $f$ such that there exists a feasible semi-orientation $\mathcal{HO}'$ entering in $(u, v)$ with $f$ and $\round{f}$ such that $\roundedmaxload{\mathcal{HO}', f, \round{f}} \leq \round{M}$ and $\roundedminload{\mathcal{HO}', f, \round{f}} \geq \round{m}$. We then write $\maxoutputflow{u}{v}{\round{f}}{\round{M}}{\round{m}} = \max(\mathcal{P})$. If $\mathcal{P} = \emptyset$, then $\maxoutputflow{u}{v}{\round{f}}{\round{M}}{\round{m}} = -\infty$.
\end{definition}

$\mininputflow{u}{v}{\round{f}}{\round{M}}{\round{m}}$ can be seen as the minimum power than should be called by the sinks of $T_v$ and produced by the sources of $T_u$. On the contrary, $\maxoutputflow{u}{v}{\round{f}}{\round{M}}{\round{m}}$ is the maximum power that can be produced by the sources of $T_u$ and called by the sinks of $T_v$. This interpretation leads intuitively to Lemma~\ref{lem:approx:4}. 

\begin{example}
	In the example of Figure~\ref{fig:halfexample}, we set $\round{m} = 0$, $\round{M} = 1$ and $u$, $v$ are the two nodes shown in the figure. There are two possible semi-orientations outgoing from $(u, v)$ satisfying the demand constraint: the one of Figure~\ref{fig:halfexample:c} and the same orientation where we reverse the arc going from $v$ to the source above. In the first case, the flow going through $(u, v)$ is $55$ and the rounded flow is around $54.991$, and, in the second case, the flow going through $(u, v)$ is $25$ and the rounded flow is around $24.9976$. Consequently, $\mininputflow{u}{v}{54.991}{1}{0} = 55$ and $\mininputflow{u}{v}{24.9976}{1}{0} = 25$. 
	
	There is no feasible orientation entering $(u, v)$ with $f = 25$ and $\round{f} = 24.9976$ as the capacity constraint is, in that case, not satisfied for $u$. The maximum value $f$ for which it is feasible is 0, thus, $\maxoutputflow{u}{v}{24.9976}{1}{0} = 0$. Similarly, $\maxoutputflow{u}{v}{54.991}{1}{0} = 0$. 
	
	As we can see, it is not possible to call a power greater than 0 to $u$ whereas $v$ needs at least 25. In other words, there is no feasible orientation where $\mathcal{O}([u, v]) = (u, v)$. 
\end{example}

\begin{lemma}
	\label{lem:approx:4}
	Let $[u, v]$ be any edge of $T$, and $\round{M}, \round{m} \in \round{L}$, there exists a feasible orientation $\mathcal{O}$ where $\round{M} \geq \roundedmaxload{\mathcal{O}}$ and $\round{m} \leq \roundedminload{\mathcal{O}}$ if and only if there exists $\round{f} \in \round{F}$ such that 
	$\mininputflow{u}{v}{\round{f}}{\round{M}}{\round{m}} \leq \maxoutputflow{u}{v}{\round{f}}{\round{M}}{\round{m}}$ or $\mininputflow{v}{u}{\round{f}}{\round{M}}{\round{m}} \leq \maxoutputflow{v}{u}{\round{f}}{\round{M}}{\round{m}}$. 
\end{lemma}
\begin{proof}
	Assuming such a feasible orientation $\mathcal{O}$ exists with $\mathcal{O}([u, v]) = (u, v)$, $f = \flow{\mathcal{O}}{u, v}$ and $\round{f} = \roundedflow{\mathcal{O}}{u, v}$. Finally, let $\mathcal{HO}$ be the orientation $\mathcal{O}$ restricted to $T_v \cup {[u, v]}$, and $\mathcal{HO}'$ be $\mathcal{O}$ restricted to $T_u \cup {[u, v]}$.
	
	Then, first, $\mathcal{HO}$ is an semi-orientation outgoing from $(u, v)$. For every edge $e \in T_v$, we have $\flow{\mathcal{O}}{e} = \flow{\mathcal{HO}}{e}$ and $\roundedflow{\mathcal{O}}{e} = \roundedflow{\mathcal{HO}}{e}$. Consequently $\mathcal{HO}$ is feasible and satisfies $\flow{\mathcal{HO}}{u, v} = f$ and $\roundedflow{\mathcal{HO}}{u, v} = \round{f}$, $\round{M} \geq \roundedmaxload{\mathcal{O}}$ and $\round{m} \leq \roundedminload{\mathcal{O}}$. Thus $\mininputflow{u}{v}{\round{f}}{\round{M}}{\round{m}} \leq f$. 
	
	Secondly $\mathcal{HO}'$ is an semi-orientation entering $(u, v)$ with $f$ and $\round{f}$. For every edge $e \in T_u$, we have $\flow{\mathcal{O}}{e} = \flow{(\mathcal{HO}', f)}{e}$ and $\roundedflow{\mathcal{O}}{e} = \roundedflow{(\mathcal{HO}', f, \round{f}))}{e}$. Consequently $\mathcal{HO}'$ is feasible and satisfies $\round{M} \geq \roundedmaxload{\mathcal{O}}$ and $\round{m} \leq \roundedminload{\mathcal{O}}$. Thus $\maxoutputflow{u}{v}{\round{f}}{\round{M}}{\round{m}} \geq f$. We then have $\maxoutputflow{u}{v}{\round{f}}{\round{M}}{\round{m}} \geq \mininputflow{u}{v}{\round{f}}{\round{M}}{\round{m}}$.

	We now assume that $\maxoutputflow{u}{v}{\round{f}}{\round{M}}{\round{m}} \geq \mininputflow{u}{v}{\round{f}}{\round{M}}{\round{m}}$ for some value $\round{f} \in \round{F}$.  We can first deduce that $\mininputflow{u}{v}{\round{f}}{\round{M}}{\round{m}} \neq +\infty$ and $\maxoutputflow{u}{v}{\round{f}}{\round{M}}{\round{m}} \neq -\infty$. By definition of $\mininputflow{u}{v}{\round{f}}{\round{M}}{\round{m}}$, there exists a feasible semi-orientation $\mathcal{HO}$ outgoing from $(u, v)$ such that $\flow{\mathcal{HO}}{u, v} = \mininputflow{u}{v}{\round{f}}{\round{M}}{\round{m}}$, $\roundedflow{\mathcal{HO}}{u, v} = \round{f}$, $\round{M} \geq \roundedmaxload{\mathcal{HO}}$ and $\round{m} \leq \roundedminload{\mathcal{HO}}$. By definition of $\maxoutputflow{u}{v}{\round{f}}{\round{M}}{\round{m}}$, there exists a feasible semi-orientation $\mathcal{HO}'$ entering $(u, v)$ with $\maxoutputflow{u}{v}{\round{f}}{\round{M}}{\round{m}}$ and $\round{f}$ such that $\round{M} \geq \roundedmaxload{\mathcal{HO}'}$ and $\round{m} \leq \roundedminload{\mathcal{HO}'}$. By Lemma~\ref{lem:approx:3}, $(\mathcal{HO}', \mininputflow{u}{v}{\round{f}}{\round{M}}{\round{m}}, \round{f})$ is also feasible.
	
	Let $\mathcal{O} = \mathcal{HO} \cup \mathcal{HO}'$. Note that, on the edges of $T_v \cup {[u, v]}$, the flow $\flowst{\mathcal{O}}$ and the rounded flow $\roundedflowst{\mathcal{O}}$ coincide respectively with $\flowst{\mathcal{HO}}$ and $\roundedflowst{\mathcal{HO}}$. In addition, on the edges of $T_u \cup {[u, v]}$, the flow $\flowst{\mathcal{O}}$ and the rounded flow $\roundedflowst{\mathcal{O}}$ coincide respectively with $\flowst{(\mathcal{HO}', \mininputflow{u}{v}{\round{f}}{\round{M}}{\round{m}}, \round{f})}$ and $\roundedflowst{(\mathcal{HO}', \mininputflow{u}{v}{\round{f}}{\round{M}}{\round{m}}, \round{f})}$.	Then $\mathcal{O}$ is feasible and $\round{M} \geq \roundedmaxload{\mathcal{O}}$ and $\round{m} \leq \roundedminload{\mathcal{O}}$.
\end{proof}

The end of the section is devoted to explaining how we can compute the auxiliary functions in polynomial time. 	
In order to prevent the technical details from obfuscating the key principles of this proof, the rest of the section is organized as follows. 
\begin{itemize}
	\item In Subsection~\ref{subsec:approx:3}, Lemma~\ref{lem:approx:7} gives the complexity of an algorithm computing the auxiliary functions $\maxoutputflowst$ and $\mininputflowst$, given some assumption on the existence of two functions $h_o$ and $h_i$ that can be computed in polynomial time,
	\item Those two functions are respectively described in Subsections~\ref{subsec:approx:4} (in Lemmas~\ref{lem:approx:5} and \ref{lem:approx:8}) and \ref{subsec:approx:5} (in Lemmas~\ref{lem:approx:6} and \ref{lem:approx:9}). In order to prove their complexities, we make another assumption on the existence of another algorithm.
	\item This algorithm is then described in Subsection~\ref{subsec:approx:6} and Lemma~\ref{lem:approx:rec}. 
\end{itemize}
	
\subsubsection{Complexity of the computation of the auxiliary functions}
\label{subsec:approx:3}

We assume $\round{M} \in \round{L}, \round{m} \in \round{L}$, $\round{f} \in \round{F}$ and $[u, v] \in T$ are given, and we want to compute $\mininputflow{u}{v}{\round{f}}{\round{M}}{\round{m}}$ and $\maxoutputflow{u}{v}{\round{f}}{\round{M}}{\round{m}}$. For readability, we denote hereinafter those values by $\mininputflowsm{u}{v}{\round{f}}$ and $\maxoutputflowsm{u}{v}{\round{f}}$ as, in the proofs, the values $\round{M}$ and $\round{m}$ never change.

\begin{lemma}
	\label{lem:approx:7}
	Given a node $u \in T$, we write $\Gamma(u)$ the neighbors of $u$. We assume that, for every edge $[u, v]$ in $T$ and every $\round{f}$, 
	\begin{itemize}
		\item there exists a function $h_o : \mathbb{N}^{2 \cdot (|\Gamma(u)| - 1) \cdot |\round{F}|} \rightarrow \mathbb{N}$ such that $$\maxoutputflowsm{u}{v}{\round{f}} = h_o\left(\left\{\maxoutputflowsm{w}{u}{\round{f'}}, \mininputflowsm{u}{w}{\round{f}'} \vert  w \in \Gamma(u) \backslash \{v\}, \round{f'} \in \round{F}\right\}\right)$$ and that can be computed in time $O(n^4 \lvert \round{F}\rvert ^4)$
		\item there exists a function $h_i : \mathbb{N}^{2 \cdot (|\Gamma(v)| - 1) \cdot |\round{F}|} \rightarrow \mathbb{N}$ such that $$\mininputflowsm{u}{v}{\round{f}} = h_i\left(\left\{\maxoutputflowsm{w}{v}{\round{f'}}, \mininputflowsm{v}{w}{\round{f}'} \vert  w \in \Gamma(v) \backslash \{u\}, \round{f'} \in \round{F}\right\}\right)$$ and that can be computed in time $O(n^4 \lvert \round{F}\rvert ^3)$
	\end{itemize}
	then we can compute $\mininputflowsm{u}{v}{\round{f}}$ and $\maxoutputflowsm{u}{v}{\round{f}}$ for all $[u, v] \in T$ and $\round{f} \in \round{F}$ in time $O(n^5 \lvert \round{F}\rvert ^5)$.
\end{lemma}
\begin{proof}
	If $v$ is a leaf, the value $\Gamma(v) \backslash \{u\}$ is empty. In that case, $\mininputflowsm{u}{v}{\round{f}} = h_i(\emptyset)$ and can be computed in time $O(n^4 \lvert \round{F}\rvert ^3)$ for each value $\round{f} \in \round{F}$. Similarly, if $u$ is a leaf, $\maxoutputflowsm{u}{v}{\round{f}} = h_o(\emptyset)$ can be computed, for each value $\round{f} \in \round{F}$, in time $O(n^4 \lvert \round{F}\rvert ^4)$ (in fact, those values can be computed in constant time, but this does not change the result of the lemma).
	
	We can then compute $\mininputflowsm{u}{v}{\round{f}}$ (respectively $\maxoutputflowsm{u}{v}{\round{f}}$) for every edge $[u, v]$ such that the height of the tree $T_v$ (respectively $T_u$) is 1. By induction on the height of the trees $T_v$ and $T_u$, we can then compute step by step from the leaves all the values $\mininputflowsm{u}{v}{\round{f}}$ and $\maxoutputflowsm{u}{v}{\round{f}}$ for all $[u, v] \in T$ and $\round{f} \in \round{F}$. At each iteration, a new value $\maxoutputflowsm{u}{v}{\round{f}}$ (respectively $\mininputflowsm{u}{v}{\round{f}}$) is evaluated in time $O(n^4 \lvert \round{F}\rvert ^4)$ (respectively $O(n^4 \lvert \round{F}\rvert ^3)$). There are $O((n - 1) \cdot \lvert \round{F}\rvert )$ possible values for $(u, v)$ and $\round{f}$, thus the lemma follows.
\end{proof}

\subsubsection{Complexity of the computation of the function $h_o$ of Lemma~\ref{lem:approx:7}}
\label{subsec:approx:4}

\begin{figure}[ht!]
	\centering
	\begin{tikzpicture}
			\tikzset{tinoeud/.style={draw, circle, minimum height=0.01cm}}
			
			\node[tinoeud, label={center:$u$}] (U) at (0,0) {};
			\node[tinoeud, label={center:$v$}] (V) at (0,-2) {};
			\node[tinoeud, label={above left:$u_1$}] (U1) at (140:4) {};
			\node[tinoeud, label={left:$u_2$}] (U2) at (170:4.5) {};
			\node[tinoeud, label={below left:$u_{\lvert J\rvert }$}] (UJ) at (200:4) {};
			\node[tinoeud, label={right:$u_{\lvert J\rvert  + 1}$}] (V1) at (40:4) {};
			\node[tinoeud, label={right:$u_{\lvert J\rvert  + 2}$}] (V2) at (10:4.5) {};
			\node[tinoeud, label={below right:$u_{p}$}] (VJ) at (-20:4) {};
			
			\draw[->,>=latex] (U1) to[out=-30, in=130] node[pos=0.3, above, sloped] {\small $g \leq \maxoutputflowsm{u_1}{u}{\round{g}} [\round{g}]$} (U);
			\draw[->,>=latex] (U2) -- node[pos=0.3, above, sloped] {\small $g \leq\maxoutputflowsm{u_2}{u}{\round{g}} [\round{g}]$} (U);
			\draw[->,>=latex] (UJ) to[out=30, in=190] node[pos=0.25, above, sloped] {\small $g \leq\maxoutputflowsm{u_{\lvert J\rvert }}{u}{\round{g}} [\round{g}]$} (U);
			\draw[->,>=latex] (U) to[out=50, in=200] node[pos=0.7, above, sloped] {\small $f_{\lvert J\rvert +1} = \mininputflowsm{u}{u_{\lvert J\rvert +1}}{\round{f_{\lvert J\rvert +1}}} [\round{f_{\lvert J\rvert +1}}]$} (V1);
			\draw[->,>=latex] (U) -- node[pos=0.8, above, sloped] {\small $f_{\lvert J\rvert +2} = \mininputflowsm{u}{u_{\lvert J\rvert +2}}{\round{f_{\lvert J\rvert +2}}} [\round{f_{\lvert J\rvert +2}}]$} (V2);
			\draw[->,>=latex] (U) to[out=-10, in=150] node[pos=0.8, above, sloped] {\small $f_p = \mininputflowsm{u}{u_{p}}{\round{f_{p}}} [\round{f_{p}}]$} (VJ);
			\draw[->,>=latex] (U) -- (V) node[midway, below right] {$f = \maxoutputflowsm{u}{v}{\round{f}} [\round{f}]$};
			
			\draw (177:1.5) node {$\vdots$};
			\draw (3:1.5) node {$\vdots$};
			
	\end{tikzpicture}
	\caption{Illustration of Lemma~\ref{lem:approx:5}, assuming $J = \{1, 2, \dots, \lvert J\rvert \}$. On each arc is written the flow followed by the rounded flow between square brackets.}
	\label{fig:approx:1}
\end{figure}
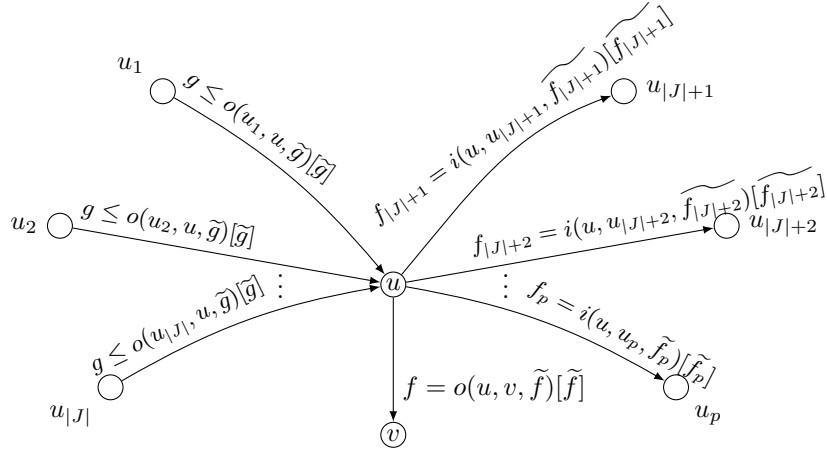

Let $u_1, u_2, \dots, u_p$ be the $p$ neighbors of $u$ in $T_u$ (every neighbor distinct from $v$). In this part, we assume $u \in W$ and $\capacity{u} = +\infty$. We will describe the formula of $\maxoutputflowsm{u}{v}{\round{f}}$ and explain how to adapt it to sources, sinks and switches with finite capacity. Without loss of generality, we also assume that $v$ has the minimum number in $\mathcal{N}$ (the numbering of the nodes of $T$) among all the neighbors of $u$.

The following lemma describes the function $h_o$ mentioned in Lemma~\ref{lem:approx:7} that, given $\maxoutputflowsm{w}{u}{\round{f'}}$ and $\mininputflowsm{u}{w}{\round{f}'}$ for every neighbor $w$ of $u$ except $v$ and every rounded flow $\round{f'}$, computes $\maxoutputflowsm{u}{v}{\round{f}}$. 

\begin{lemma}
	\label{lem:approx:5}
	
	Let $J  \subseteq \llbracket 1; p \rrbracket$, with $J \neq \emptyset$, and $\round{f_j}  \in \round{F}$ for  $j \not\in J$. 
	
	We set $\round{g} = \oplus((\round{f}, \round{f_j})_{j \not\in J}, \lvert J\rvert )$ where $\round{f_j} = 0$ if $j \in J$, and
	
	$$f(J, (\round{f_j})_{j \not\in J}) = \begin{cases}
		\lvert J\rvert  \cdot \min\limits_{j \in J} \maxoutputflowsm{u_j}{u}{\round{g}} - \sum\limits_{j \not\in J} \mininputflowsm{u}{u_j}{\round{f_j}} &\text{if this value is non-negative}\\
		-\infty &\text{otherwise }
	\end{cases}$$
	
	then
	$$\maxoutputflowsm{u}{v}{\round{f}} = 
	\max\limits_{\substack{J  \subseteq \llbracket 1; p \rrbracket\\J \neq \emptyset}}
	\max\limits_{\substack{\round{f_j}  \in \round{F}\\ j \not\in J}}
	f(J, (f_j)_{j \not\in J})$$	 
	
\end{lemma}
\begin{example}
	\label{example:oexample}
\begin{figure}[!ht]
	\centering
	\begin{tikzpicture}
		\clip (-3.75,-1.5) rectangle (0.75, 2);
		
		\node[source] [label=right:$$]  [label=above:$$]  (S1) at (-3,1) {};
		\node[source] [label=left:$s_2$] [label=above:$20/20$] (S2) at (0,1) {};
		
		\node[switch] (W1) at (-3,0) {};
		\node[switch] [label=above:$v$] (W12) at (-1.5,0) {};
		\node[switch] [label=right:$u$] (W2) at (0,0) {};
		
		\node[puits]  (P1) at (-3,-1) {};
		\node[puits] [label=left:$15$] [label=right:$p_2$] (P2) at (0,-1) {};
		\node[puits]  (P3) at (-1.5,1) {};
		
		\draw[-, >=latex, dotted] (S1) -- (W1);
		\draw[-, >=latex, dotted] (S1) -- (P3);
		\draw[-, >=latex] (S2) -- (W2);
		
		\draw[-, >=latex, dotted] (W1) -- (P1);
		\draw[-, >=latex, dotted] (W12) -- (W1);
		\draw[<-, >=latex] (W12) -- (W2);
		\draw[-, >=latex] (W2) -- (P2);
	\end{tikzpicture}
	\caption{Illustation of Example~\ref{example:oexample}.}
	\label{fig:oexample}
\end{figure}
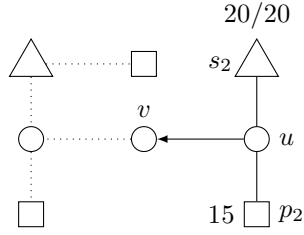

We use the instance of Figure~\ref{fig:oexample} to illustrate the Lemma. We assume that $\varepsilon = 1$ to simplify the calculations of rounding values. We also assume that $\tilde{m} = 0.2$ and $\tilde{M} = 0.5$. We would like to compute $\maxoutputflowsm{u}{v}{4}$ that is the greatest power that can be sent from $u$ to $v$ if the rounded flow between them is $4$. Intuitively, this value is $5$ as the source $s_2$ can produce at most $20$ and must send $15$ to $p_2$. 

Now lets compute the formula of Lemma~\ref{lem:approx:5}. We get:

$$\maxoutputflowsm{u}{v}{4} = \max \begin{cases}
	\maxoutputflowsm{s_2}{u}{\round{g}} - \mininputflowsm{u}{p_2}{\round{f_2}} &\text{ where } \round{g} = \oplus((4, \round{f_2}), 1)\\
	2 \cdot \min(\maxoutputflowsm{s_2}{u}{\round{g}}, \maxoutputflowsm{p_2}{u}{\round{g}}) &\text{ where } \round{g} = \oplus((4), 2)\\
	\maxoutputflowsm{p_2}{u}{\round{g}} - \mininputflowsm{u}{s_2}{\round{f_1}} &\text{ where } \round{g} = \oplus((4, \round{f_1}), 1)
\end{cases}$$

The value $\maxoutputflowsm{p_2}{u}{\round{g}}$ is necessarily $-\infty$. Indeed, there is not feasible semi-orientation where the arc $(p_2, u)$ is directed toward $u$ as the demand constraint would not be satisfied for $p_2$. This leave us with

$$\maxoutputflowsm{u}{v}{4} = \maxoutputflowsm{s_2}{u}{\round{g}} - \mininputflowsm{u}{p_2}{\round{f_2}} \text{ where } \round{g} = \oplus((4, \round{f_2}), 1)$$

The value $\mininputflowsm{u}{p_2}{\round{f_2}}$ is $+\infty$ if $\round{f_2}$ is not the rounded value of $\power{p_2}$, that is $a(15) = 8$ if $\varepsilon = 1$, and $\mininputflowsm{u}{p_2}{8}$ is $\power{p_2} = 15$. The value $\round{g}$ is then $\oplus((4, 8), 1) = a(8  + a(4)) = a(12) = 8$. We search for $\maxoutputflowsm{s_2}{u}{8}$. A rounded flow for $s_2$ of $8$ is valid as its rounded load would be $8/20$, which is between $\round{m}$ and $\round{M}$. As $s_2$ cannot produce more than $20$, $\maxoutputflowsm{s_2}{u}{8}$ is 20. 

Thus we get $\maxoutputflowsm{u}{v}{4} = 20 - 15 = 5$ as predicted.
\end{example}
\begin{remark}[Intuition]
	Figure~\ref{fig:approx:1} illustrates Lemma~\ref{lem:approx:5}. Given a subset $J$ and rounded flows $(\round{f_j})_{j \not\in J}$, the value $f(J, (\round{f_j})_{j \not\in J})$ corresponds to the maximum flow that can go through $u$ to $v$ if, for every $j \not\in J$, the rounded flow in $(u, u_j)$ is $\round{f_j}$ and if, for every $j \in J$, all the edges $(u_j, u)$ are directed toward $u$. In those arcs, assuming the rounded flow in $(u, v)$ is $\round{f}$, the rounded flow is necessarily $\round{g}$.
	
	$\lvert J\rvert  \cdot \min_{j \in J} \maxoutputflowsm{u_j}{u}{\round{g}}$ is the maximum value that can be equitably distributed over all the input arcs of $u$. And $\sum_{j \not\in J} \mininputflowsm{u}{u_j}{\round{f_j}}$ is the value requested by the neighbors directed from $u$ other than $v$, consequently, the maximum value that can go from $u$ to $v$ is the difference.
	
	Note that $J$ must be not empty in order to direct at least one arc toward $u$ to satisfy the demand constraint.
\end{remark}
\begin{proof}

	Let $ma = 
	\max\limits_{\substack{J  \subseteq \llbracket 1; p \rrbracket\\J \neq \emptyset}}
	\max\limits_{\substack{\round{f_j}  \in \round{F}\\ j \not\in J}}
	f(J, (f_j)_{j \not\in J})$. 
	
	We assume that $\maxoutputflowsm{u}{v}{\round{f}} \neq -\infty$. Then, by definition of $\maxoutputflowsm{u}{v}{\round{f}}$, there exists a feasible semi-orientation $\mathcal{HO}$ entering $(u, v)$ with $f = \maxoutputflowsm{u}{v}{\round{f}}$ and $\round{f}$ such that $\roundedmaxload{(\mathcal{HO}, f, \round{f})} \leq \round{M}$ and $\roundedminload{(\mathcal{HO}, f, \round{f})} \geq \round{m}$. Let $J$ be the set of indexes $j \in \llbracket 1;p \rrbracket$ such that $\mathcal{HO}([u, u_j]) = (u_j, u)$. The set $J$ is not empty, otherwise the demand constraint is not satisfied for $u$.  We write $g = \flow{(\mathcal{HO}, f, \round{f})}{u_j, u}$, $\round{g} = \roundedflow{(\mathcal{HO}, f, \round{f})}{u_j, u}$, and $f_j = \flow{(\mathcal{HO}, f, \round{f})}{u, u_j}$, $\round{f_j} = \roundedflow{(\mathcal{HO}, f, \round{f})}{u, u_j}$ for $j \not\in J$. By definition of the rounded flow, and because $v$ has the minimum number in $\mathcal{N}$ among all the neighbors of $u$, $\oplus((\round{f}, \round{f_j})_{j \not\in J}, \lvert J\rvert ) = \round{g}$. Finally, by definition of $\flowst{(\mathcal{HO}, f, \round{f})}$, $f =  \lvert J\rvert  \cdot g - \sum\limits_{j \not\in J} f_j$.
	
	Let $T_{u_j}$ be the subtree of $T_u$ containing $u_j$ in the forest $T_u \backslash [u, u_j]$, let $j \in J$, and let $\mathcal{HO}_j$ be the semi-orientation entering $(u_j, u)$ that equals $\mathcal{HO}$ restricted to $T_{u_j}$ then $\mathcal{HO}_j$ is feasible with $g$ and $\round{g}$. In addition, $\roundedmaxload{\mathcal{HO}_j, g, \round{g}} \leq \roundedmaxload{\mathcal{HO}, f, \round{f}} \leq \round{M}$ and $\roundedminload{\mathcal{HO}_j, g, \round{g}} \geq \roundedminload{\mathcal{HO}, f, \round{f}} \geq \round{m}$ thus $\flow{(\mathcal{HO}, f, \round{f})}{u_j, u} \leq \maxoutputflowsm{u_j}{u}{\round{g}}$. As, for every edge $[u_j, u]$, $\flow{(\mathcal{HO}, f, \round{f})}{u_j, u} = g$, we have, $g \leq \min\limits_{j \in J} \maxoutputflowsm{u_j}{u}{\round{g}}$. We can similarly show that, for $j \not\in J$, $f_j \geq \mininputflowsm{u}{u_j}{\round{f_j}}$. Consequently, $f =  \lvert J\rvert  \cdot g - \sum\limits_{j \not\in J} f_j \leq \lvert J\rvert  \cdot \min\limits_{j \in J} \maxoutputflowsm{u_j}{u}{\round{g}}  - \sum\limits_{j \not\in J} \mininputflowsm{u}{u_j}{\round{f_j}} \leq ma$. Thus  $\maxoutputflowsm{u}{v}{\round{f}} \leq ma$. As $\maxoutputflowsm{u}{v}{\round{f}} \neq -\infty$, we can also deduce that $ma \neq -\infty$. 
	
	We now assume that $ma \neq -\infty$. Let $J \neq \emptyset$ and $\round{f_j}$ for $j \not\in J$, such that $f(J, (\round{f_j})_{j \not\in J}) = ma = \lvert J\rvert  \cdot \min\limits_{j \in J} \maxoutputflowsm{u_j}{u}{\round{g}}  - \sum\limits_{j \not\in J} \mininputflowsm{u}{u_j}{\round{f_j}}$ with the value $\round{g}$ being $\oplus((\round{f}, \round{f_j})_{j \not\in J}, \lvert J\rvert )$ (as $J \neq \emptyset$, this equality is defined). We build an semi-orientation $\mathcal{HO}$ entering $(u, v)$ satisfying $ma \leq \maxoutputflowsm{u}{v}{\round{f}}$. To do so, we set $\mathcal{HO}([u_j, u]) = (u_j, u)$ if $j \in J$ and $(u, u_j)$ otherwise. As $ma \neq -\infty$, we know that, if $j \in J$, then $\maxoutputflowsm{u_j}{u}{\round{g}} \neq -\infty$ and, if $j \not\in J$, $\mininputflowsm{u}{u_j}{\round{f_j}} \neq +\infty$. Thus, for every $j \in J$, there exists a feasible semi-orientation $\mathcal{HO}_j$ entering $(u_j, u)$ with $\maxoutputflowsm{u_j}{u}{\round{g}}$ and $\round{g}$ such that $\roundedmaxload{\mathcal{HO}_j, \maxoutputflowsm{u_j}{u}{\round{g}}, \round{g}} \leq \round{M}$ and $\roundedminload{\mathcal{HO}_j, \maxoutputflowsm{u_j}{u}{\round{g}}, \round{g}} \geq \round{m}$. Similarly, for every $j \not\in J$, there exists a feasible halt-orientation $\mathcal{HO}_j$ outgoing from $(u, u_j)$ such that $\roundedmaxload{\mathcal{HO}_j} \leq \round{M}$, $\roundedminload{\mathcal{HO}_j} \geq \round{m}$,  $\flow{\mathcal{HO}_j}{u, u_j} = \mininputflowsm{u}{u_j}{\round{f_j}}$ and $\roundedflow{\mathcal{HO}_j}{u, u_j} = \round{f_j}$. Let then $\mathcal{HO}$ be the union of all those semi-orientations plus $\mathcal{HO}([u, v]) = (u, v)$. 
	
	If $(\mathcal{HO}, ma,\round{f})$ is feasible, then, because we have $\roundedmaxload{\mathcal{HO}, ma, \round{f}} \leq \round{M}$ and $\roundedminload{\mathcal{HO}, ma, \round{f}} \geq \round{m}$, then $ma \leq \maxoutputflowsm{u}{v}{\round{f}}$ and $\maxoutputflowsm{u}{v}{\round{f}} \neq -\infty$. First, note that, because $J \neq \emptyset$, the node $u$ has an entering arc and the demand constraint is then satisfied for $u$. It is satisfied for every other node as all the semi-orientations $\mathcal{HO_j}$ are feasible. Secondly, for all $j \not\in J$, $\flow{(\mathcal{HO}, ma, \round{f})}{u, u_j} = \flow{\mathcal{HO}_j}{u, u_j} = \mininputflowsm{u}{u_j}{\round{f_j}}$ and $\roundedflow{(\mathcal{HO}, ma, \round{f})}{u, u_j} = \round{f_j}$. For all $j \in J$, the flow $\flow{(\mathcal{HO}, ma, \round{f})}{u_j, u} = \frac{1}{\lvert J\rvert } \cdot (ma + \sum\limits_{j \not\in J} \flow{(\mathcal{HO}, ma, \round{f})}{u, u_j}) = \min\limits_{j \in J} \maxoutputflowsm{u_j}{u}{\round{g}}$ and the rounded flow $\roundedflow{(\mathcal{HO}, f, \round{f})}{u_j, u} = \oplus((\round{f},  \round{f_j})_{j \not\in J}, \lvert J\rvert ) = \round{g}$ (again because $v$ has the minimum number in $\mathcal{N}$ among all the neighbors of $u$). Finally, for all $j \in J$, $(\mathcal{HO}_j, \maxoutputflowsm{u_j}{u}{\round{g}}, \round{g})$ is feasible. According to Lemma~\ref{lem:approx:3}, then $(\mathcal{HO}_j, \min\limits_{j \in J} \maxoutputflowsm{u_j}{u}{\round{g}}, \round{g})$ is also feasible. Consequently $(\mathcal{HO}, ma, \round{f})$ is feasible.
\end{proof}

The following lemma proves that we can compute the formula given in Lemma~\ref{lem:approx:5} in polynomial time. This lemma makes a hypothesis that is proven in Lemma~\ref{lem:approx:rec}.  

\begin{lemma}
	\label{lem:approx:8}\hfill
 \begin{itemize}
 	\item If we have access to $\maxoutputflowsm{u_j}{u}{\round{f}}$ and $\mininputflowsm{u}{u_j}{\round{f}}$ in $O(1)$ for every $j \in \llbracket 1; p \rrbracket$ and $\round{f} \in \round{F}$,
 	\item and if, given $d \in \llbracket 1; p \rrbracket$, $k \in \llbracket 1; p - d + 1 \rrbracket$ and $\tilde{g} \in \tilde{F}$, we can compute $$\min\limits_{\substack{J  \subseteq \llbracket 1; p \rrbracket\\ \lvert J\rvert  = d\\\min(J) = k}}
 	\min\limits_{\substack{\round{f_j}  \in \round{F}\\ j \not\in J\\ \oplus((\round{f},  \round{f_j})_{j \not\in J}, d) = \round{g}}}
 	\sum\limits_{j \not\in J} \mininputflowsm{u}{u_j}{\round{f_j}}$$ in time $O(n^2 \lvert \round{F}\rvert ^3)$,
 \end{itemize}
 then we can compute $\maxoutputflowsm{u}{v}{\round{f}}$ in $O(n^4\lvert \round{F}\rvert ^4)$.
\end{lemma}
\begin{proof}
	
	By Lemma~\ref{lem:approx:5},
	
	\begin{align*}
		\maxoutputflowsm{u}{v}{\round{f}} &= 
		\max\limits_{\substack{J  \subseteq \llbracket 1; p \rrbracket\\ J \neq \emptyset}}
		\max\limits_{\substack{\round{f_j}  \in \round{F}\\ j \not\in J}}
		f(J, (f_j)_{j \not\in J})\\
		\maxoutputflowsm{u}{v}{\round{f}} &= \max\limits_{d = 1}^p \max\limits_{\substack{J  \subseteq \llbracket 1; p \rrbracket\\ \lvert J\rvert  = d}}
		\max\limits_{\substack{\round{f_j}  \in \round{F}\\ j \not\in J}}
		f(J, (f_j)_{j \not\in J})\\
		\maxoutputflowsm{u}{v}{\round{f}} &= \max\limits_{d = 1}^p \max\limits_{k = 1}^{p-d + 1} \max\limits_{\substack{J  \subseteq \llbracket 1; p \rrbracket\\ \lvert J\rvert  = d\\\min(J) = k}}
		\max\limits_{\substack{\round{f_j}  \in \round{F}\\ j \not\in J}}
		f(J, (f_j)_{j \not\in J})\\
		\maxoutputflowsm{u}{v}{\round{f}} &= \max\limits_{d = 1}^p \max\limits_{k = 1}^{p-d + 1} 
		\max\limits_{\round{g} \in \round{F}} 		\max\limits_{\substack{J  \subseteq \llbracket 1; p \rrbracket\\ \lvert J\rvert  = d\\\min(J) = k}}
		\max\limits_{\substack{\round{f_j}  \in \round{F}\\ j \not\in J\\ \oplus((\round{f},  \round{f_j})_{j \not\in J}, d) = \round{g}}}
		f(J, (f_j)_{j \not\in J})\\
		\intertext{If, for every $\round{g} \in \round{F}$, we write $o_1(\round{g}), o_2(\round{g}), \dots, o_p(\round{g})$ the $p$ values $\maxoutputflowsm{u_j}{u}{\round{g}}$ sorted in ascending order}
		\maxoutputflowsm{u}{v}{\round{f}} &= \max\limits_{d = 1}^p \max\limits_{k = 1}^{p-d + 1}
		\max\limits_{\round{g} \in \round{F}} 		\max\limits_{\substack{J  \subseteq \llbracket 1; p \rrbracket\\ \lvert J\rvert  = d\\\min(J) = k}}
		\max\limits_{\substack{\round{f_j}  \in \round{F}\\ j \not\in J\\ \oplus((\round{f},  \round{f_j})_{j \not\in J}, d) = \round{g}}}
		d \cdot o_k(\round{g}) - \sum\limits_{j \not\in J} \mininputflowsm{u}{u_j}{\round{f_j}}
	\end{align*}
	
	Given $d$, $k$ and $\round{g}$, we can compute $\max\limits_{\substack{J  \subseteq \llbracket 1; p \rrbracket\\ \lvert J\rvert  = d\\\min(J) = k}}
	\max\limits_{\substack{\round{f_j}  \in \round{F}\\ j \not\in J\\ \oplus((\round{f},  \round{f_j})_{j \not\in J}, d) = \round{g}}}
	d \cdot o_k(\round{g}) - \sum\limits_{j \not\in J} \mininputflowsm{u}{u_j}{\round{f_j}}$ to solve the problem, and this can be computed in time $O(n^2 \lvert \round{F}\rvert ^3)$ by hypothesis. If the obtained value is negative, we return $-\infty$.
	
	As a conclusion $\maxoutputflowsm{u}{v}{\round{f}}$ can be computed in time $O(n^4 \lvert \round{F}\rvert ^4)$ by enumerating all the possible values of  $d$, $k$ and $\round{g}$, and by returning the maximum result obtained with all the triplets.
\end{proof}

\begin{remark}
	We can adapt the formula of Lemma~\ref{lem:approx:5} to other cases.
	\begin{itemize}
		\item If the numbering of $v$ in $\mathcal{N}$ is not the smallest among the numbers of all the neighbors of $u$, the list $(\round{f},  \round{f_j})_{j \not\in J}$ should be sorted in the formula of $\round{g}$.
		\item If $u$ is a switch with a finite capacity, we must replace $\lvert J\rvert  \cdot \min\limits_{j \in J} \maxoutputflowsm{u_j}{u}{\round{g}}$ by $\min(\capacity{u}, \lvert J\rvert  \cdot \min\limits_{j \in J} \maxoutputflowsm{u_j}{u}{\round{g}})$ in the formula of the function $f(J, (\round{f_j})_{j \not\in J})$.
		\item If $u$ is a source, then we can have $J = \emptyset$ and we must replace $\lvert J\rvert  \cdot \min\limits_{j \in J} \maxoutputflowsm{u_j}{u}{\round{g}}$ by $\min((\lvert J\rvert  + 1) \cdot \production{u}, (\lvert J\rvert  + 1) \cdot \min\limits_{j \in J} \maxoutputflowsm{u_j}{u}{\round{g}})$; in addition, $\oplus((\round{f},  \round{f_j})_{j \not\in J}, \lvert J\rvert )$ must be replaced by $\oplus((\round{f},  \round{f_j})_{j \not\in J}, \lvert J\rvert  + 1)$ in the formula of $\round{g}$; and finally, we have to replace  $\max\limits_{\round{g}  \in \round{F}}$ by $\max\limits_{\round{g} \in \round{F}, \round{m}\leq \frac{\round{g}}{\production{u}} \leq \round{M}}$ in the formula of $\maxoutputflowsm{u}{v}{\round{f}}$.
		\item If $u$ is a sink, we must replace the formula of $\round{g}$ by $\oplus((\power{u}, \round{f},  \round{f_j})_{j \not\in J}, \lvert J\rvert )$ and replace $- \sum\limits_{j \not\in J} \mininputflowsm{u}{u_j}{\round{f_j}}$ by $- \sum\limits_{j \not\in J} \mininputflowsm{u}{u_j}{\round{f_j}} - \power{u}$ in the formula of $f(J, (\round{f_j})_{j \not\in J})$.
	\end{itemize}
	
\end{remark}

\subsubsection{Complexity of the computation of the function $h_i$ of Lemma~\ref{lem:approx:7}}
\label{subsec:approx:5}

\begin{figure}[ht!]
	\centering
	\begin{tikzpicture}
		\tikzset{tinoeud/.style={draw, circle, minimum height=0.01cm}}
		
		\node[tinoeud, label={center:$v$}] (U) at (0,0) {};
		\node[tinoeud, label={center:$u$}] (V) at (0,-2) {};
		\node[tinoeud, label={left:$v_1$}] (U1) at (140:4) {};
		\node[tinoeud, label={left:$v_2$}] (U2) at (170:4.5) {};
		\node[tinoeud, label={left:$v_{\lvert J\rvert }$}] (UJ) at (200:4) {};
		\node[tinoeud, label={right:$v_{\lvert J\rvert  + 1}$}] (V1) at (40:4) {};
		\node[tinoeud, label={right:$v_{\lvert J\rvert  + 2}$}] (V2) at (10:4.5) {};
		\node[tinoeud, label={below right:$v_{p}$}] (VJ) at (-20:4) {};
		
		\draw[->,>=latex] (U1) to[out=-30, in=130] node[pos=0.3, above, sloped] {\small $f \leq \maxoutputflowsm{v_1}{v}{\round{f}} [\round{f}]$} (U);
		\draw[->,>=latex] (U2) -- node[pos=0.3, above, sloped] {\small $f \leq \maxoutputflowsm{v_2}{v}{\round{f}} [\round{f}]$} (U);
		\draw[->,>=latex] (UJ) to[out=30, in=190] node[pos=0.25, above, sloped] {\small $f \leq \maxoutputflowsm{v_{\lvert J\rvert }}{v}{\round{f}} [\round{f}]$} (U);
		\draw[->,>=latex] (U) to[out=50, in=200] node[pos=0.7, above, sloped] {\small $f_{\lvert J\rvert +1} = \mininputflowsm{v}{v_{\lvert J\rvert +1}}{\round{f_{\lvert J\rvert +1}}} [\round{f_{\lvert J\rvert +1}}]$} (V1);
		\draw[->,>=latex] (U) -- node[pos=0.7, above, sloped] {\small $f_{\lvert J\rvert +2} = \mininputflowsm{v}{v_{\lvert J\rvert +2}}{\round{f_{\lvert J\rvert +2}}} [\round{f_{\lvert J\rvert +2}}]$} (V2);
		\draw[->,>=latex] (U) to[out=-10, in=150] node[pos=0.8, above, sloped] {\small $f_p = \mininputflowsm{u}{v_{v}}{\round{f_{p}}} [\round{f_{p}}]$} (VJ);
		\draw[<-,>=latex] (U) -- (V) node[midway, below right] {$f = \mininputflowsm{u}{v}{\round{f}} [\round{f}]$};
		
		\draw (177:1.5) node {$\vdots$};
		\draw (3:1.5) node {$\vdots$};
		
	\end{tikzpicture}
	\caption{Illustration of Lemma~\ref{lem:approx:6}. On each arc is written the flow followed by the rounded flow between square brackets.}
	\label{fig:approx:2}
\end{figure}
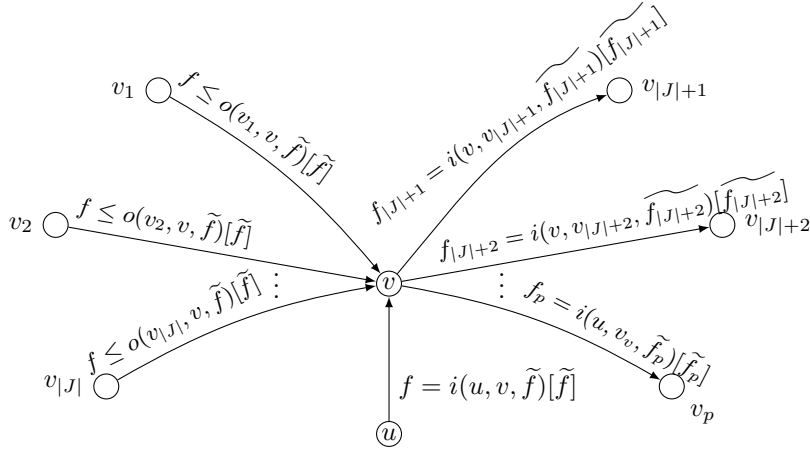

Let $v_1, v_2, \dots, v_p$ be the $p$ neighbors of $v$ in $T_v$ (every neighbor distinct from $u$). In this part, we assume $v \in W$ and $\capacity{v} = +\infty$. We will describe the formula of $\mininputflowsm{u}{v}{\round{f}}$ and explain how to adapt it to sources, sinks and switches with finite capacity. 

The following lemma describes the function $h_i$ mentioned in Lemma~\ref{lem:approx:7} that, given $\maxoutputflowsm{w}{v}{\round{f'}}$ and $\mininputflowsm{v}{w}{\round{f}'}$ for every neighbor $w$ of $v$ except $u$ and every rounded flow $\round{f'}$, computes $\mininputflowsm{u}{v}{\round{f}}$. 

\begin{lemma}
	\label{lem:approx:6}
	Let $J  \subseteq \llbracket 1; p \rrbracket$, and $\round{f_j}  \in \round{F}$ for $j \not\in J$ such that $\round{f} = \oplus((\round{f_j})_{j \not\in J}, \lvert J\rvert  + 1)$.
	
	We set
	$$f(J, (\round{f_j})_{j \not\in J}) = \begin{cases}\frac{1}{\lvert J\rvert  + 1}\sum\limits_{j \not\in J} \mininputflowsm{v}{v_j}{\round{f_j}} &\text{ if this value is lower than } \min\limits_{j \in J} \maxoutputflowsm{v_j}{v}{\round{f}}\\
		+\infty & \text{ otherwise}
	\end{cases}$$
	
	If $J = \emptyset$, we set $\min\limits\limits_{j \in J} \maxoutputflowsm{v_j}{v}{\round{f}} = +\infty$.

	Then
	
	$$\mininputflowsm{u}{v}{\round{f}} = \min\limits_{J \subseteq \llbracket 1; p \rrbracket}\min\limits_{\substack{\round{f_j}  \in \round{F}\\ j \not\in J\\\oplus((\round{f_j})_{j \not\in J}, \lvert J\rvert  + 1) = \round{f}}} f(J, (\round{f_j})_{j \not\in J}) $$
	
\end{lemma}
\begin{example}
	\label{example:iexample}
	\begin{figure}[!ht]
		\centering
		\begin{tikzpicture}
			\clip (-3.75,-1.5) rectangle (0.75, 2);
			
			\node[source] [label=right:$$]  [label=above:$$]  (S1) at (-3,1) {};
			\node[source] [label=left:$s_2$] [label=above:$7.5/20$] (S2) at (0,1) {};
			
			\node[switch] (W1) at (-3,0) {};
			\node[switch] [label=above:$u$] (W12) at (-1.5,0) {};
			\node[switch] [label=right:$v$] (W2) at (0,0) {};
			
			\node[puits]  (P1) at (-3,-1) {};
			\node[puits] [label=left:$15$] [label=right:$p_2$] (P2) at (0,-1) {};
			\node[puits]  (P3) at (-1.5,1) {};
			
			\draw[-, >=latex, dotted] (S1) -- (W1);
			\draw[-, >=latex, dotted] (S1) -- (P3);
			\draw[->, >=latex] (S2) -- (W2);
			
			\draw[-, >=latex, dotted] (W1) -- (P1);
			\draw[-, >=latex, dotted] (W12) -- (W1);
			\draw[->, >=latex] (W12) -- (W2);
			\draw[-, >=latex] (W2) -- (P2);
		\end{tikzpicture}
		\caption{Illustation of Example~\ref{example:iexample}.}
		\label{fig:iexample}
	\end{figure}
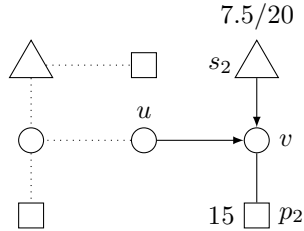
	
	We use the instance of Figure~\ref{fig:iexample} to illustrate the Lemma. We assume that $\varepsilon = 1$ to simplify the calculations of rounding values. We also assume that $\tilde{m} = 0.2$ and $\tilde{M} = 0.5$. We would like to compute $\mininputflowsm{u}{v}{4}$ that is the lower power that can be called by $v$ to $u$ if the rounded flow between them is $4$. Intuitively, this value is $7.5$ as the sink $p_2$ calls $15$ to $v$, which can either equally calls $7.5$ to each of $u$ and $s_2$ or calls $15$ to $u$ alone. 
	
	Now lets compute the formula of Lemma~\ref{lem:approx:6}. We get:
	
	$$\mininputflowsm{u}{v}{4} = \min \begin{cases}
		\frac{1}{2} \cdot \mininputflowsm{v}{p_2}{\round{f_1}} &\text{ if it is lower than }  \maxoutputflowsm{s_2}{v}{4} \text{ where } 4 = \oplus((\round{f_1}), 2)\\
		\frac{1}{2} \cdot \mininputflowsm{v}{s_2}{\round{f_2}} &\text{ if it is lower than }  \maxoutputflowsm{p_2}{v}{4} \text{ where } 4 = \oplus((\round{f_2}), 2)\\
		\mininputflowsm{v}{p_2}{\round{f_1}} + \mininputflowsm{v}{s_2}{\round{f_2}} &\text{ if it is lower than } +\infty \text{ where } 4 = \oplus((\round{f_1}, \round{f_2}), 1)\\
		\frac{1}{3} \cdot 0 &\text{ if it is lower than } \min(\maxoutputflowsm{p_2}{v}{4}, \maxoutputflowsm{s_2}{v}{4}) 
	\end{cases}$$

	The value $\maxoutputflowsm{p_2}{v}{4}$ is necessarily $-\infty$. Indeed, there is not feasible semi-orientation where the arc $(p_2, v)$ is directed toward $v$ as the demand constraint would not be satisfied for $p_2$. This leaves us with the first and third case of the formula.
	
	The value $\mininputflowsm{v}{s_2}{\round{f_2}}$ is $+\infty$ if $\round{f_2} \neq 0$, indeed, no power is called by $s_2$ if $(v, s_2)$ is directed toward $s_2$. However, $\mininputflowsm{v}{s_2}{0}$ is also infinite as the source would have a rounded load of 0 and this value is not between $\round{m}$ and $\round{M}$. Consequently,
	
	$$\mininputflowsm{u}{v}{4} = \frac{1}{2} \cdot \mininputflowsm{v}{p_2}{\round{f_1}} \text{ if it is lower than }  \maxoutputflowsm{s_2}{v}{4} \text{ where } 4 = \oplus((\round{f_1}), 2)$$
	
	The value $\mininputflowsm{v}{p_2}{\round{f_1}}$ is $+\infty$ if $\round{f_1}$ is not the rounded value of $\power{p_2}$, that is $a(15) = 8$ if $\varepsilon = 1$, and $\mininputflowsm{v}{p_2}{8} = \power{p_2} = 15$. We have $\oplus((8), 2) = 4$. The rounded load of $s_2$ is $4/20$ if the rounded flow of $s_2$ is 4, which is between $\round{m}$ and $\round{M}$. Thus $\maxoutputflowsm{s_2}{v}{4}$ is 20. Finally, $\frac{1}{2} \mininputflowsm{v}{p_2}{8} \leq \maxoutputflowsm{s_2}{v}{4}$. 
	
	Therefore, we get $\mininputflowsm{u}{v}{4} = \frac{1}{2} \cdot \mininputflowsm{v}{p_2}{8} = 7.5$ as predicted.
\end{example}
\begin{remark}[Intuition]
	Figure~\ref{fig:approx:2} illustrates Lemma~\ref{lem:approx:6}. Given a subset $J$, and rounded flows $(\round{f_j})_{j \not\in J}$, the value $f(J, (\round{f_j})_{j \not\in J})$ is the minimum flow requested to $v$ if, for every $j \not\in J$, the rounded flow in $(v, v_j)$ is $\round{f_j}$ and if, for every $j \in J$, all the edges $(v_j, v)$ are directed toward $v$. Assuming the rounded flow in $(u, v)$ is $\round{f}$, as it is equal in all the input arcs of $v$, the rounded flow is then also $\round{f}$ in $(v_j, v)$, and that value should equal $\oplus((\round{f_j})_{j \not\in J}, \lvert J\rvert  + 1)$.
	
	As the flow is also equal in all the input arcs, the minimum flow requested to $v$ cannot be more than $\min_{j \in J} \maxoutputflowsm{v_j}{v}{\round{f}}$.
	
	Note that, as $v$ is a switch, it looks tempting to set $J = \llbracket 1; n \rrbracket$ so that the flow in $(u, v)$ is 0 (and thus minimum). However, three problems may occur: first, it does not work if $\round{f} \neq 0$, secondly, this solution may not satisfy the choice of $\round{m}$ and $\round{M}$, and thirdly, we may have $\maxoutputflowsm{v_j}{v}{\round{f}} = -\infty$.
\end{remark}
\begin{proof}
	Let $mi = \min\limits_{J \subseteq \llbracket 1; p \rrbracket}\min\limits_{\substack{\round{f_j}  \in \round{F}\\ j \not\in J\\\oplus((\round{f_j})_{j \not\in J}, \lvert J\rvert  + 1) = \round{f}}} f(J, (\round{f_j})_{j \not\in J})$.
	
	If $\mininputflowsm{u}{v}{\round{f}} \neq +\infty$, then there exists a feasible semi-orientation $\mathcal{HO}$ outgoing from $(u, v)$ such that $\roundedmaxload{\mathcal{HO}} \leq \round{M}$, $\roundedminload{\mathcal{HO}} \geq \round{m}$, $\flow{\mathcal{HO}}{u, v} = \mininputflowsm{u}{v}{\round{f}}$ and $\roundedflow{\mathcal{HO}}{u, v} = \round{f}$. We hereinafter denote by $f$ the value $\mininputflowsm{u}{v}{\round{f}}$. Let $J$ be the set of indexes $j \in \llbracket 1;p \rrbracket$ such that $\mathcal{HO}([v, v_j]) = (v_j, v)$. We write $\round{f_j} = \roundedflow{\mathcal{HO}}{v, v_j}$ for $j \not\in J$. By definition of the rounded flow, $\oplus((\round{f_j})_{j \not\in J}, \lvert J\rvert  + 1) = \round{f} = \roundedflow{\mathcal{HO}}{v_j, v}$ for $j \in J$. Finally, by definition of the flow, for $j \in J$, $\round{f} = \roundedflow{\mathcal{HO}}{v_j, v} = 1/(\lvert J\rvert  + 1)\sum_{j \not\in J} f_j$.
	
	Let $T_{v_j}$ be the subtree of $T_v$ containing $v_j$ in the forest $T_v \backslash [v, v_j]$, let $j \in J$, and let $\mathcal{HO}_j$ be the semi-orientation entering $(v_j, v)$ that equals $\mathcal{HO}$ restricted to $T_{v_j}$ then $\mathcal{HO}_j$ with $f$ and $\round{f}$ is feasible. In addition, $\roundedmaxload{\mathcal{HO}_j, f, \round{f}} \leq \roundedmaxload{\mathcal{HO}} \leq \round{M}$ and $\roundedminload{\mathcal{HO}_j, f, \round{f}} \geq \roundedminload{\mathcal{HO}} \geq \round{m}$ thus $f \leq \maxoutputflowsm{v_j}{v}{\round{f}}$, and consequently $f \leq \min_{j \in J} \maxoutputflowsm{v_j}{v}{\round{g}}$. We can similarly show that, for $j \not\in J$, $f_j \geq \mininputflowsm{v}{v_j}{\round{f_j}}$. Finally $f = 1/(\lvert J\rvert  + 1)\sum_{j \not\in J} f_j$, consequently, $1/(\lvert J\rvert  + 1)\sum_{j \not\in J} \mininputflowsm{v}{v_j}{\round{f_j}} \leq \min\limits_{j \in J} \maxoutputflowsm{u_j}{u}{\round{f}}$. Therefore $f(J, (\round{f_j})_{j \not\in J})$ is not infinite and equals $f$; and then $mi \leq f(J, (\round{f_j})_{j \not\in J}) \leq f$. As $f =\mininputflowsm{u}{v}{\round{f}}$, we conclude that $mi \leq \mininputflowsm{u}{v}{\round{f}}$. We can also deduce that $mi \neq +\infty$.
	
	We now assume $mi \neq +\infty$. Let $J$ and $\round{f_j}$, for $j \not\in J$, such that $\oplus((\round{f_j})_{j \not\in J}, \lvert J\rvert  + 1) = \round{f}$ and $f(J, (\round{f_j})_{j \not\in J}) = mi = 1/(\lvert J\rvert  + 1)\sum_{j \not\in J} \mininputflowsm{v}{v_j}{\round{f_j}}$. We build an semi-orientation $\mathcal{HO}$ outgoing from $(u, v)$ satisfying $mi \geq \mininputflowsm{u}{v}{\round{f}}$. To do so, we set $\mathcal{HO}([v_j, v]) = (v_j, v)$ if $j \in J$ and $(v, v_j)$ otherwise. As $mi \neq +\infty$ then $f(J, (\round{f_j})_{j \not\in J}) \leq \min\limits_{j \in J} \maxoutputflowsm{v_j}{v}{\round{f}}$. Thus if $j \in J$, $\maxoutputflowsm{v_j}{v}{\round{f}} \neq -\infty$ and, if $j \not\in J$, $\mininputflowsm{v}{v_j}{\round{f_j}} \neq +\infty$. Consequently, there exists, for every $j \in J$, a feasible semi-orientation $\mathcal{HO}_j$ entering $(v_j, v)$ with $\maxoutputflowsm{v_j}{v}{\round{f}}$ and $\round{f}$ such that $\roundedmaxload{\mathcal{HO}_j, \maxoutputflowsm{v_j}{v}{\round{f}}, \round{f}} \leq \round{M}$ and $\roundedminload{\mathcal{HO}_j, \maxoutputflowsm{v_j}{v}{\round{f}}, \round{f}} \geq \round{m}$. Similarly, for every $j \not\in J$, there exists an semi-orientation $\mathcal{HO}_j$ outgoing from $(v, v_j)$ such that $\roundedmaxload{\mathcal{HO}_j} \leq \round{M}$, $\roundedminload{\mathcal{HO}_j} \geq \round{m}$, $\flow{\mathcal{HO}_j}{v, v_j} = \mininputflowsm{v}{v_j}{\round{f_j}}$ and $\roundedflow{\mathcal{HO}_j}{v, v_j} = \round{f_j}$. Let $\mathcal{HO}$ be the union of all those semi-orientations plus $\mathcal{HO}([u, v]) = (u, v)$. 
	
	If $\mathcal{HO}$ is feasible, $\flow{\mathcal{HO}}{u, v} = mi$ and $\roundedflow{\mathcal{HO}}{u, v} = \round{f}$, then, because we have $\roundedmaxload{\mathcal{HO}, f, \round{f}} \leq \round{M}$ and $\roundedminload{\mathcal{HO}, f, \round{f}} \geq \round{m}$, then $mi \geq \mininputflowsm{u}{v}{\round{f}}$ and then $\mininputflowsm{u}{v}{\round{f}} \neq +\infty$. Note firstly that the demand constraint is satisfied for $u$ as it has an input arc and it is satisfied for all the other nodes as all the semi-orientations are feasible. For all $j \not\in J$, $\flow{\mathcal{HO}}{v, v_j} = \flow{\mathcal{HO}_j}{v, v_j} = \mininputflowsm{v}{v_j}{\round{f_j}}$ and $\round{F}_{\mathcal{HO}}(v, v_j) = \round{F}_{\mathcal{HO}_j}(v, v_j) = \round{f_j}$. For all $j \in J$, $\flow{\mathcal{HO}}{u, v} = \flow{\mathcal{HO}}{v_j, v} = 1/(\lvert J\rvert  + 1) \cdot \sum_{j \not\in J} \flow{\mathcal{HO}}{v, v_j} = mi$. In addition, we made the assumption that $\oplus((\round{f_j})_{j \not\in J}, \lvert J\rvert  + 1) = \round{f}$, hence $\roundedflow{\mathcal{HO}}{u, v} = \roundedflow{\mathcal{HO}}{v_j, v} = \round{f}$. Finally, for every $j \in J$, $(\mathcal{HO}_j, \maxoutputflowsm{v_j}{v}{\round{f}}, \round{f})$ is feasible. As $mi \leq \min_{j \in J} \maxoutputflowsm{v_j}{v}{\round{f}}$ then, by Lemma~\ref{lem:approx:3}, the triplet $(\mathcal{HO}_j, \flow{\mathcal{HO}}{v_j, v}, \round{f})$ is also feasible. As a consequence $\mathcal{HO}$ is feasible.
\end{proof}

The following lemma proves we can compute the formula given in Lemma~\ref{lem:approx:6} in polynomial time. This lemma makes a hypothesis that is proven in Lemma~\ref{lem:approx:rec}.  

\begin{lemma}
	\label{lem:approx:9}\hfill
	\begin{itemize}
		\item If we have access to $\maxoutputflowsm{v_j}{v}{\round{f}}$ and $\mininputflowsm{v}{v_j}{\round{f}}$ in $O(1)$ for every $j \in \llbracket 1; p \rrbracket$ and $\round{f} \in \round{F}$,
		\item and if, given $d \in \llbracket 0; p \rrbracket$ and $k \in \llbracket 1; p - d + 1 \rrbracket$, we can compute $$\min\limits_{\substack{J \subseteq \llbracket 1; p \rrbracket\\\lvert J\rvert  = d\\\min(J) = k}}\min\limits_{\substack{\round{f_j}  \in \round{F}\\ j \not\in J\\\oplus((\round{f_j})_{j \not\in J}, d + 1) = \round{f}}}
		\sum\limits_{j \not\in J} \mininputflowsm{v}{v_j}{\round{f_j}}$$ in $O(n^2 \lvert \round{F}\rvert ^3)$ (where $\min(\emptyset)$ is arbitrarily set to $p + 1$ so that the formula is defined for $d = 0$)
	\end{itemize}
then we can compute $\mininputflowsm{u}{v}{\round{f}}$ in $O(n^4\lvert \round{F}\rvert ^3)$.
\end{lemma}
\begin{proof}
	By Lemma~\ref{lem:approx:6}
	
	\begin{align*}
		\mininputflowsm{u}{v}{\round{f}} &= \min\limits_{J \subseteq \llbracket 1; p \rrbracket}\min\limits_{\substack{\round{f_j}  \in \round{F}\\ j \not\in J\\\oplus((\round{f_j})_{j \not\in J}, \lvert J\rvert  + 1) = \round{f}}} f(J, (\round{f_j})_{j \not\in J})\\
		\mininputflowsm{u}{v}{\round{f}} &= \min\limits_{d = 0}^p \min\limits_{\substack{J \subseteq \llbracket 1; p \rrbracket\\\lvert J\rvert  = d}}\min\limits_{\substack{\round{f_j}  \in \round{F}\\ j \not\in J\\\oplus((\round{f_j})_{j \not\in J}, d + 1) = \round{f}}} f(J, (\round{f_j})_{j \not\in J})\\
		\mininputflowsm{u}{v}{\round{f}} &= \min\limits_{d = 0}^p \min\limits_{k = 1}^{p-d +1} \min\limits_{\substack{J \subseteq \llbracket 1; p \rrbracket\\\lvert J\rvert  = d\\\min(J) = k}}\min\limits_{\substack{\round{f_j}  \in \round{F}\\ j \not\in J\\\oplus((\round{f_j})_{j \not\in J}, d + 1) = \round{f}}}f(J, (\round{f_j})_{j \not\in J})\\
		\intertext{If we write $o_1(\round{f}), o_2(\round{f}), \dots, o_p(\round{f})$ the $p$ values $\maxoutputflowsm{v_j}{v}{\round{f}}$ sorted in ascending order}
		\mininputflowsm{u}{v}{\round{f}} &= \min\limits_{d = 0}^p \min\limits_{k = 1}^{p-d +1} \min\limits_{\substack{J \subseteq \llbracket 1; p \rrbracket\\\lvert J\rvert  = d\\\min(J) = k}}\min\limits_{\substack{\round{f_j}  \in \round{F}\\ j \not\in J\\\oplus((\round{f_j})_{j \not\in J}, d + 1) = \round{f}}} \begin{cases}\frac{1}{d + 1}\sum\limits_{j \not\in J} \mininputflowsm{v}{v_j}{\round{f_j}} & \\\text{ if this value is lower than } o_k(\round{f})\\
			+\infty \text{ otherwise}
		\end{cases}
	\end{align*}

	Given $d$ and $k$, we can compute $\min\limits_{\substack{J \subseteq \llbracket 1; p \rrbracket\\\lvert J\rvert  = d\\\min(J) = k}}\min\limits_{\substack{\round{f_j}  \in \round{F}\\ j \not\in J\\\oplus((\round{f_j})_{j \not\in J}, d + 1) = \round{f}}}
	\sum\limits_{j \not\in J} \mininputflowsm{v}{v_j}{\round{f_j}}$ to solve the problem, and this can be computed in time $O(n^2 \lvert \round{F}\rvert ^3)$ by hypothesis. We then compare the returned value to $o_k(\round{f})$, if it is lower, we return it, otherwise we return $+\infty$. 
	
	As a conclusion $\mininputflowsm{u}{v}{\round{f}}$ can be computed in $O(n^4 \lvert \round{F}\rvert ^3)$ by enumerating all the possible values of $d$ and $k$, and by returning the minimum result obtained with all the couples.
\end{proof}

\begin{remark}
	We can adapt the formula of Lemma~\ref{lem:approx:6} to other cases:
	\begin{itemize}
		\item if $v$ is a switch with a finite capacity, we must replace $\min\limits_{j \in J} \maxoutputflowsm{v_j}{v}{\round{g}}$ by $\min(\capacity{v}, \min\limits_{j \in J} \maxoutputflowsm{v_j}{v}{\round{g}})$ in the formula of the function $f(J, (\round{f_j})_{j \not\in J})$.
		\item if $v$ is a source, we must replace $\frac{1}{\lvert J\rvert  + 1}\sum\limits_{j \not\in J} \mininputflowsm{v}{v_j}{\round{f_j}}$ if this value is lower than $\min\limits_{j \in J} \maxoutputflowsm{v_j}{v}{\round{f}}$ by $\frac{1}{\lvert J\rvert  + 2}\sum\limits_{j \not\in J} \mininputflowsm{v}{v_j}{\round{f_j}}$ if this value is lower than $\min(\production{v}, \min\limits_{j \in J} \maxoutputflowsm{v_j}{v}{\round{f}})$; in addition, $\oplus((\round{f_j})_{j \not\in J}, \lvert J\rvert  + 1) = \round{f}$ must be replaced by $\oplus((\round{f_j})_{j \not\in J}, \lvert J\rvert  + 2) = \round{f}$ and, finally, we must set $\mininputflowsm{u}{v}{\round{f}}$ to $+\infty$ if we do not check the inequalities $\round{m} \leq \round{f} / \production{v} \leq \round{M}$
		\item if $v$ is a sink, we must replace the formula $\oplus((\round{f_j})_{j \not\in J}, \lvert J\rvert  + 1) = \round{f}$ by $\oplus((\power{v}, \round{f_j})_{j \not\in J}, \lvert J\rvert  + 1) = \round{f}$; and replace $\sum\limits_{j \not\in J} \mininputflowsm{v}{v_j}{\round{f_j}}$ by $\sum\limits_{j \not\in J} \mininputflowsm{v}{v_j}{\round{f_j}} + \power{v}$ in the formula of $f(J, (\round{f_j})_{j \not\in J})$.
	\end{itemize}
\end{remark}

\subsubsection{Hypothesis of Lemmas~\ref{lem:approx:8} and \ref{lem:approx:9}}
\label{subsec:approx:6}

This final part proves the hypothesis made in the two lemmas. Let $\round{f}, \round{g} \in \round{F}$, $u$ be a node of $T$ and $u_1, u_2, \dots, u_p$ be a subset of the neighbors of $u$ in $T$. Let $d \in \llbracket 0; p \rrbracket$, $d' \in \llbracket d; d+1 \rrbracket$ and $k \in \llbracket 1; p - d + 1 \rrbracket$ if $d \neq 0$ and $k = p + 1$ otherwise. We arbitrarily set $\min(\emptyset) = p + 1$. We compute the following value:

$$\Xi = \min\limits_{\substack{J  \subseteq \llbracket 1; p \rrbracket\\ \lvert J\rvert  = d\\\min(J) = k}}
\min\limits_{\substack{\round{f_j}  \in \round{F}\\ j \not\in J\\ \oplus((\round{f}, \round{f_j}_{j \not\in J}), d') = \round{g}}}
\sum\limits_{j \not\in J} \mininputflowsm{u}{u_j}{\round{f_j}}$$

\begin{remark}
	\label{rem:arrondi:2}
	Recall the remark on page~\pageref{rem:arrondi} interrogating why the formula operator $\oplus$ introduces many rounding errors. This hypothesis is where this choice intervenes as it makes Lemma~\ref{lem:approx:rec} true. 
\end{remark}

\begin{lemma}
	\label{lem:approx:rec}
	Assuming we can access to $\mininputflowsm{u}{u_j}{\round{f}}$ in $O(1)$ for every $j \in \llbracket 1; p \rrbracket$ and $\round{f} \in \round{F}$, then we can compute $\Xi$ in $O(n^2 \lvert \round{F}\rvert ^3)$.
\end{lemma}
\begin{proof}
	$\Xi$ can be equivalently seen as the search of $p$ couples $(\round{f_j}, b_j)_{p \in \llbracket 1; p \rrbracket}$ where $b_j$ is a boolean determining whether $j \in J$ or not and where: 
	
	\begin{enumerate}
		\item $b_k$ is true and there exist $d-1$ other true $b_j$ with $j > k$; every other boolean is false;
		\item if $b_j$ is true then $\round{f_j} = 0$;
		\item $\round{g} = \oplus((\round{f}, \round{f_1}, \round{f_2}, \dots, \round{f_{p}}), d')$;
		\item $\sum\limits_{\substack{j = 1\\b_j \text{ is false}}}^p \mininputflowsm{u}{u_j}{\round{f_j}}$ is minimum.
	\end{enumerate}
Note that if $d = 0$ and $k = p + 1$, then all the booleans are false. 
	
	We use a dynamic programming algorithm to compute the following recursive function $H$. We set $H(i, \round{f}', s)$ as the function that equals the minimum value $\sum\limits_{\substack{j = 1\\b_j \text{ is false}}}^i \mininputflowsm{u}{u_j}{\round{f_j}}$ that can be obtained with $i$ couples $(\round{f_1}, b_1), (\round{f_2}, b_2), \dots, (\round{f_i}, b_i)$ satisfying
	\begin{enumerate}
		\item if $i \geq k$ then $b_k$ is true and there exist $s-1$ other true $b_j$ with $j > k$; every other boolean is false; if $i < k$ then $s = 0$ and all the booleans are false. 
		\item if $b_j$ is true then $\round{f_j} = 0$
		\item $\round{f}' = \oplus((\round{f}, \round{f_1}, \round{f_2}, \dots, \round{f_{i}}), d')$
	\end{enumerate}
	
	If no such pairs exists then $H(i, \round{f'}, s) = +\infty$. By computing $H(p, \round{g}, d)$, we obtain the desired result. 
	
	We then have the following recurrence formula. 
	
	$$ H(0, \round{f'}, s) = 0 \text{ if } s = 0 \text{ and } \round{f'} = a(\frac{\round{f}}{d'}) \text{, and } +\infty \text{ otherwise}$$
	
	Indeed, if $i = 0$, then the sum we want to minimize is null, no boolean is true and $\round{f}' = \oplus((\round{f}), d')$.
	
	$$H(i, \round{f'}, s) = \min \begin{cases}
		H(i - 1, \round{f"}, s) - \mininputflowsm{u}{u_i}{\round{f_i}} \text{ if } i \neq k \text{ and } a(\round{f"} + \frac{\round{f_i}}{d'}) = \round{f'}\\
		H(i - 1, \round{f'}, s - 1) \text{ if } i \geq k\\
		+\infty
	\end{cases}$$

	Note that, if $d = 0$ and $k = p + 1$, then the second case never occurs, meaning that $s$ should equal $0$ to reach the terminal state $H(0, \round{f'}, 0)$. 
	
	The first case sets $b_i$ to false and the second one to true. With a dynamic programming algorithm, this formula can be computed in time $O(n^2 \lvert \round{F}\rvert ^3)$ as there are $n$ possible values for $i$ and $s$, and $\lvert \round{F}\rvert $ for $\round{f}'$; and as the first part of the formula is computed in time $O(\lvert \round{F}^2\rvert )$ by enumerating all the couples $(\round{f"}, \round{f_i})$ and the second one in $O(1)$.
\end{proof}

\subsubsection{Main result}
	
\begin{theorem}
	There exists an FPTAS for \maxminm and an FPTAS with absolute ratio for \maxminm and \minr running in time $O(n^7 \lvert \round{F}\rvert ^7)$ where $\lvert \round{F}\rvert  = O(\frac{1}{\varepsilon'} \cdot (n^3\log(n) + n^2\log(\Pi)))$.
\end{theorem}
\begin{proof}
	By Lemma~\ref{lem:approx:4}, \ref{lem:approx:7}, \ref{lem:approx:8}, \ref{lem:approx:9} and 
	\ref{lem:approx:rec}, given $\round{M}, \round{m} \in \round{L}$, we can decide whether there exists a feasible orientation $\mathcal{O}$ where $\round{M} \geq \roundedmaxload{\mathcal{O}}$ and $\round{m} \leq \roundedminload{\mathcal{O}}$ or not in $O(n^5 \lvert \round{F}\rvert ^5)$. By computing this for every value of $\round{M}$ and $\round{m}$, in $O(n^5 \lvert \round{F}\rvert ^5 \lvert \round{L}\rvert ^2)$, we can get a feasible orientation $\mathcal{O}_r$ minimizing $\roundedloadreserve{\mathcal{O}}$ or a feasible orientation $\mathcal{O}_m$ maximizing $\roundedminload{\mathcal{O}}$. By Lemmas~\ref{lem:approx:maxminm} and \ref{lem:approx:minr}, $\mathcal{O}_r$ is an approximate solution with absolute ratio $\varepsilon'$, and $\mathcal{O}_m$ is an approximate solution with absolute ratio $\varepsilon'$ and relative ratio $(1 - \varepsilon')$.
	
	Consequently, this algorithm is an FPTAS for \maxminm and an FPTAS with absolute ratio for \maxminm and \minr. The complexity follows from the fact that $\lvert \round{L}\rvert  = O(n \lvert \round{F}\rvert )$ and from Lemma~\ref{lem:nbflots}.
\end{proof}

\section{Conclusion}

In this paper, we have studied the theoretical complexity and approximability of three new optimization flow problems where a constraint of equality is added at the entrance of each node, and whose objective is to return the orientation that balance as much as possible the load of the sources of the network. A reasonable continuation of this work is to implement and evaluate the performance of the polynomial algorithm for \minmaxm and the FPTAS algorithms for \maxminm and \minr on real electrical distribution networks, to determine their runtime efficiency and the resilience of the proposed configurations.

With this in mind, further work is needed first to improve the complexity of the algorithms. For example, it may not be necessary to enumerate all the values of $\round{M}$, $\round{m}$ when searching for a solution optimizing the rounded load reserve with Lemma~\ref{lem:approx:4}. The computation of the functions $\mininputflowst$ and $\maxoutputflowst$ should also be improved. 

Another continuation of the work would be to adapt the algorithms to networks that are \emph{almost} trees, such as graphs with small treewidth or with small cyclomatic number. 

\section*{Declarations}

\begin{itemize}
	\item This work is partially supported by the ANR project Labcom HYPHES  21-LCV1-0002.
	\item The authors have no relevant financial or non-financial interests to
	disclose.
	\item No extra material was used to produce this paper
\end{itemize}

\bibliographystyle{acm}
\bibliography{edla}

\end{document}